\documentclass[review]{elsarticle}

\usepackage{lineno,hyperref}
\usepackage[dvipsnames]{xcolor}
\usepackage{algorithm}% http://ctan.org/pkg/algorithms
\usepackage{algpseudocode}% http://ctan.org/pkg/algorithmicx
\usepackage{algorithmicx}
\usepackage{amsthm}
\usepackage{amsmath}

\newcommand{\note}[2]{{\color{red}[[\textbf{#1:}#2]]}}

\DeclareMathOperator*{\argmin}{arg\,min}

\newtheorem{theorem}{Theorem}

\newtheorem{lemma}[theorem]{Lemma}
\newtheorem{definition}[theorem]{Definition}

\journal{Computers and Security}

\begin{document}

\begin{frontmatter}

\title{Attack Graph Obfuscation}

\author{Hadar Polad}
\ead{poladh@post.bgu.ac.il}
\author{Rami Puzis}
\ead{puzis@bgu.ac.il}
\author{Bracha Shapira}
\ead{bshapira@bgu.ac.il}
\address{Software and Information Systems Enginnering, \\Ben-Gurion University of the Negev}

\begin{abstract}
Before executing an attack, adversaries usually explore the victim's network in an attempt to infer the network topology and identify vulnerabilities in the victim's servers and personal computers. 
Falsifying the information collected by the adversary post penetration may significantly slower lateral movement and increase the amount of noise  generated within the victim's network. 
We investigate the effect of fake vulnerabilities within a real enterprise network on the attacker performance. 
We use the attack graphs to model the path of an attacker making its way towards a target in a given network. 
We use combinatorial optimization in order to find the optimal assignments of fake vulnerabilities.
We demonstrate the feasibility of our deception-based defense by presenting results of experiments with a large scale real network. We show  that adding fake vulnerabilities forces the adversary to invest a significant amount of effort, in terms of time and exploitability cost.
\end{abstract}


\end{frontmatter}

%\linenumbers

\section{Introduction}
\label{sec:intro}

Protecting a network is always a difficult task because attackers constantly explore new ways to penetrate security systems by exploiting their vulnerabilities. These vulnerabilities often go unpatched, due to  lack of resources, negligence, or a variety of other reasons.
\\Although network professionals have offered various versions of the attack process over the years, today the general anatomy of the attack process is thought to be comprised of five steps \cite{murphy2010application}:
\begin{enumerate}
\item Reconnaissance
\item Scanning
\item Gaining access
\item Maintaining access
\item Covering tracks
\end{enumerate}
Some networking professionals estimate that an adversary routinely spends up to 95\% of its time planning an attack, while only spending the remaining 5\% on execution  \cite{kewley2001dynamic}. 
During the reconnaissance step, the attacker attempts to gather as much information about the designated network as possible, including network topology, operating systems and applications, and unpatched vulnerabilities.
While doing so, the adversary generates traffic on the network, making itself more vulnerable  for detection\cite{huber2011host}. 
\\In this research, we try to sabotage the reconnaissance and scanning steps of the attack process by obfuscating the information acquired by the adversary. While making the attacker repeat steps 1 and 2 repeatedly, after failing to achieve step 3.
\par It is well known, that attackers rely upon the ability to accurately identify the operating system and services running on the network in order to plan and execute successful attacks \cite{murphy2010application}.
\\The desire to mislead a possible attacker underlies this research aiming to explore the possibilities of obfuscating the information acquired by an adversary. This has been achieved by adding fake vulnerabilities that distract the attacker and contribute to the erroneous construction of an attack path. Misleading the attacker with false information can set the attacker on a path that will deplete its resources, increase the likelihood of detection due to the increased activity, and keep the attacker away from essential targets.
We hypothesize that adding fake vulnerabilities will cause the attackers to perform additional activities while attempting to achieve their goals.
\\In this study we assume that the attacker will choose the path with the lowest total cost of the resulting attack graph. In addition, the attack graph is constructed from the information known to the adversary. 
Furthermore, we assume that the attacker knows the structure of the given network, and the vulnerabilities in each host.
\\In this research, we make the following contributions:
\begin{itemize}

\item We present a new defense strategy for protecting enterprise networks. This method utilizes attack graphs for modeling all possible attack plans in a given network. Fake vulnerabilities are then added to hosts in the network in order to make it harder for an adversary to reach its goal in the target network.
\item In contrast to  other studies that use small synthetic networks, this study examines the impact of our algorithm on a real enterprise network, and demonstrate the solution of the above challenge for a real organization.
\item We gathered a collection of guidelines for fake vulnerabilities placement.\\
This study considered the fact that when a layer of deception is added to a host in a network, it inhibits the network's routine activity.
Therefore, the user, i.e., the enterprise aiming to protect its network, should decide about the desired level of security it wishes to apply and the  resources it can and will provide for the task of protecting its network.\\
Another consideration pertains to the fake information provided to attackers. If the fake information is naive or poorly chosen, the attacker may immediately become suspicious and assume that the responses obtained are deceptive \cite{rowe2007defending}.
In our research, we add the deceptive information carefully and sensibly, in such a way that it cannot be easily detected by an attacker. While applying deceptive information to a specific host, we assure that it is consistent with the environment and with other information that can be concluded.\\
Further more, we provide an efficient and effective  algorithm for assignment of fake vulnerabilities in the network.

\item We formulate our problem as an AI challenge. We are looking for a group of assignments which maximizes the adversary's attack total cost, while it tries to find the path with minimal cost. Actually, the attacker is trying to solve a planning (AI) problem  while our method searches for a hard instance of this planning problem. 
In a general saying, we formulate a search problem whose goal is to find a hard instance for an AI algorithm. 

\item We present an admissible but still feasible heuristic solution for the defined problem. The admissibility property assures the optimal solution to our problem. 

\end{itemize}

%%%%%%%%%%%%%%%%%%%%%%%%%%%%%%%%%%%%%%%%%%%%%%%%%%%%%%%%%%%%%%%%
%% Background
%%%%%%%%%%%%%%%%%%%%%%%%%%%%%%%%%%%%%%%%%%%%%%%%%%%%%%%%%%%%%%%%

\section{Background on attack graphs} 
\label{sec:ag}
Attack graphs are data structures used to model the possible paths an attacker could use to achieve its goal within a specific target network.
%In previous studies two major attack graph models are presented: \textbf{scenario graphs} \cite{sheyner2002automated,sheyner2004scenario}  and \textbf{logical graphs}  \cite{ou2006scalable}. 
The earliest attack graphs \cite{sheyner2002automated,sheyner2004scenario} were constructed manually by Red Teams and could not scale to large networks.  
Later, technological advancement and the introduction of logical attack graphs~\cite{ou2006scalable} made attack graph generation more scalable and comprehensible to the human user. 
In general, attack graph generation requires the complete network connectivity map and the list of existing vulnerabilities in the network hosts~\cite{khaitan2011finding}.

The task of collecting the necessary data about the vulnerabilities and the network structure  must be automated.  
The former can be collected using either one of the existing vulnerability scanners, for example Nessus \cite{nessus} or openVAS \cite{OpenVAS}. 
Various tools, such as NMAP~\cite{lyon2009nmap}, can aid in the network topology assessment. 
However, accurate assessment of connectivity within large organizations is still an open problem due to firewall rules, intrusion detection systems (IDS), lack of documented switch and router configurations, etc.

\subsection{Attacker model}
In this study we assume that the adversary's goal is to obtain designated privileges in the victim's network, while minimizing attack cost. 
The adversary also tries to minimize its footprint in the attacked network and to avoid redundant / superfluous actions that create additional noise.  
%its use of resources and making as little noise in the network as possible. 
The more noise the adversary generates, the higher are the chances  that the malicious activity will be detected by an IDS. 
Reducing the consumption of resources is an objective shared by both defenders and attackers alike.

This led us to conclude that the best approach for modeling the attacker would follow the following assumptions:
\begin{itemize}
\item The attacker considers, on every step, the resources it will need to spend, during the attack. 
This describes the cost of each vertex in the attack graph.
\item The attacker aims to choose the path to its goal in the targeted network with the lowest cost.
\end{itemize}
We will construct the attack graph considering the above, while using the information the attacker can get. In this research we assume the attacker has all the information about the network - topology and vulnerabilities at each host.

\subsection{Logical attack graphs} 
\label{sec:lag}
We adopted the definition of logical attack graphs proposed by Ou et al.~\cite{ou2006scalable}. 
\begin{definition}[Logical Attack Graph]
\label{def:ag}
A logical attack graph is a five-tupple $G=(Np, Ne, Nc, E, g)$, where $Np$, $Ne$, and $Nc$ are disjoint sets of privilege nodes, exploit nodes, and configuration nodes respectively. 
$E$ is the set of directed graph edges 
$E \subseteq \left(Np \times Ne\right) \cup \left(Ne \times \left(Np \cup Nc\right) \right)$. 
$g \in Np$ is the attacker's goal.
%Finally, the labeling function $L$ maps nodes to their semantic representation. 
\end{definition}

Privilege nodes represent the various assets and privileges that can be acquired on the host machines (by using an exploit). 
Configuration nodes mainly represent vulnerabilities but can also include various configurations such as the operating system and services installed on the host.

Consider for example the attack graph depicted in Figure~\ref{fig:exampleLogical}. 
Consider the exploit node $e1$.
It can only be activated if configuration $c1$ is in place and the privileges $p2$ and $p3$ were acquired by the attacker. 
Once activated, $e1$ grants the privilege $p1$.  
Alternatively, privilege $p1$ can be acquired using the exploit $e2$. 

%The semantics of a logical attack graph are defined as follows:
%For every exploit node $e$, let $P$ be $e\lq s$ parent node and $C$ be the set of $e\lq s$ child nodes, then $(\land L(C) \Rightarrow L(P))$ is an instantiation of interaction rule $L(R)$.  \cite{ou2006scalable}\cite{chatterjee2014dragon}
%\par Due to its visualization and scalability (\textbf{polynomial} in the size of the network) which are suitable for enterprise networks, we use the logical attack graph model in our research.\\
%A logical attack graph can be generated automatically and relatively easily, for example, using the \textbf{MulVAL} framework  \cite{ou2005mulval}. %(see section \ref{par:mulval}).
%A small example of a logical attack graph is presented in Figure~\ref{fig:exampleLogical}.
%%%%%%%%%%%%%%%% begin Figure %%%%%%%%%%%%%%%%%%%
\begin{figure}
  \centering
  \includegraphics[width=.6\linewidth]{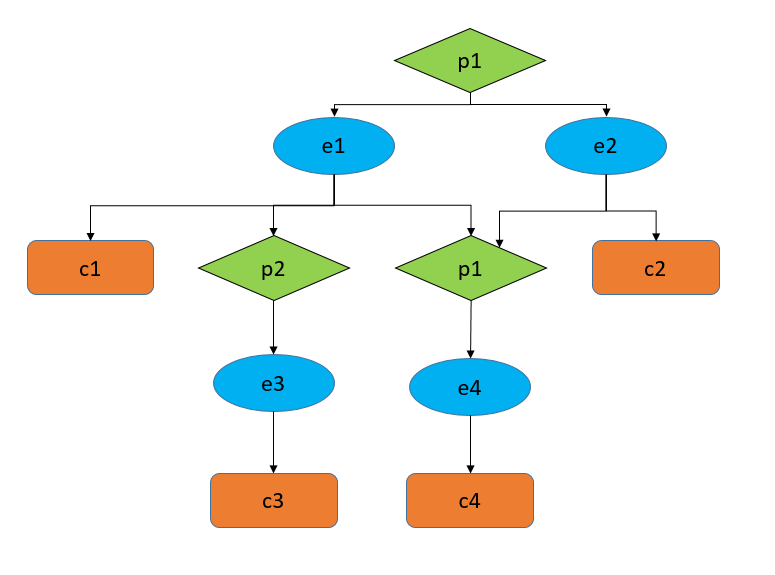}
  \caption{Example of logical attack graph. Green diamonds are privilege nodes, Blue ovals are exploit nodes, Orange rectangles are configuration nodes. Arrows point to the logical requirement of each node.}
  \label{fig:exampleLogical}
\end{figure}
%%%%%%%%%%%%%%%% end Figure %%%%%%%%%%%%%%%%%%%
%% I would explain the example

\subsection{Attack graph generation tools:}
\label{sec:ag-gen}
%Due to scalability and visualization challenges introduced by attack graphs, generating an attack graph manually is impractical and even impossible, so an automatic tool is needed. We will present some noteworthy generating tools:
There are multiple tools for generating logical attack graphs. 
The most notable are Network Security Planning Architecture (NetSPA)~\cite{artz2002netspa,ingols2006practical}, Graphical Attack Graph and Reachability Network Evaluation Tool (GARNET)~ \cite{williams2008garnet}, and multi-host, multi-stage Vulnerability Analysis Language (MulVAL)~ and framework \cite{ou2005mulval} which enables automatic and relatively easy generation of attack graphs.

\label{par:mulval}
We use MulVAL for generating the attack graphs. 
%After the initial success of attack graphs, and the apparent drawbacks which include poor-scalability to large-enterprise networks in model checking techniques, Ou et al proposed MulVAL \cite{ou2005mulval} "Multi-host, multi-stage Vulnerability Analysis Language". 
MulVAL is a framework for modeling the interaction of the attacker with vulnerabilities and network configurations. 
The inputs to MulVAL are:
\begin{itemize}
\item  Known exploits and their effects.
\item  Hosts configurations - including installed software.
\item  Network configurations - including VLANs and firewall tables.
\item  Principals - including users and their permissions.
\item  Interaction rules for the different parts of the system.
\item  Policy - what accesses does the system wants to permit.
\end{itemize}

%Although configuring all these parameters sounds like a daunting task, in practice, 
Most of the configurations required for running MulVAL can be imported from automatic tools or public sources. 
For example, known exploits could be imported from NVD (National Vulnerability Database)\cite{NVD}. 
Host configurations can be imported from automatic network scanners such as Nessus \cite{nessus}. 
Once the configurations are present, they are converted to facts, and rules - the basic building blocks of MulVAL's reasoning system. 
From this point, the facts and rules can be directly loaded into a Prolog environment and executed. 
MulVAL uses XSB Prolog dialect~\cite{sagonas1994xsb} because it averts from re-computation of formerly calculated facts. 
This results in a polynomial run time complexity for the attack simulation phase.

Empirical tests show that in practice, MulVAL is very efficient in computing attack simulation, even in scenarios including thousands of hosts. 
In a typical MulVAL run, it will check if there is a trace that will result in a policy violation using the facts and rules defined earlier. 
The result of such run is an attack graph, containing all the possible paths from the initial configuration to the goal privileges.
MulVAL is a popular tool in recent studies related to attack graphs. %bracha - we need to add references for that, at least many studeis on attack grpahs that use mulval
Its advantages are noticeable, and therefore we use it in the current study.

\subsection{Exploitation cost} 
\label{sec:cve}
Common vulnerabilities have been compiled and listed in a system operated by the MITRE Corporation and the U.S. National Vulnerability Database~\cite{NVD}. 
Each vulnerability is tagged with a unique CVE (Common Vulnerabilities and Exposures) identifier and has a CVSS (Common Vulnerabilities Scoring System) score, which is an open industry standard for assessing the severity of computer system security vulnerabilities.
CVSS scores are based on two subscores:
\begin{itemize}
\item Impact Subscore - reflects the direct consequence of a successful exploit and represents the consequence to the impacted component.
\item Exploitability Subscore - reflects the ease and technical means by which the vulnerability can be exploited.
\end{itemize}
%This is an open information that can be accessed using standard APIs.  
%this information can be extracted is gathered in an xml files in the NVD website \cite{NVD}, which is free for use.
In this study we use the exploitability subscore as a cost weight on configuration nodes that represent vulnerabilities.  
\begin{definition}[Vulnerability exploit cost]
\label{def:costv}
Let $v\in N_c$ be a vulnerability node, then $C_v\in [0..1]$ denotes the normalized cost of exploiting $v$. 
\end{definition}
The normalized cost of vulnerability exploit can be obtained by dividing the exploitability subscore by 10 (CVSS v2.0) or by 3.9 (CVSS v3.0).%bracha - I did not understand the 3.9 and 10 - is it something that I miss or should be explained?? 
%Adding cost to the action nodes by their matching exploitability subscore (see section \ref{sec:cve}), which reflects the ease and technical means by which the vulnerability can be exploited.
Future enhancements using alternative cost functions are discussed in Section~\ref{sec: Conclusions and future work}.

\section{Obfuscating the Attack Graph} 
\label{sec:method}
In this section we describe the idea of obfuscating an attack graph by adding fake vulnerabilties and approaches for solving it.  
First we formulate the attack graph obfuscation as a search problem and discuss the random assignment of fake vulnerabilties~\cite{polad2017attack} as a baseline approach. 
Then we employ Depth First Brand and Bound algorithm to find the optimal solution. 
In order to speed up the search process we define a few admissible heuristics and a node ordering function to guide the search. 

\subsection{Problem definition}
\label{sec:problem}

Assume an attack graph $G$ according to Definition~\ref{def:ag}. 
All the following notations are defined with respect to an attack graph $G$ unless specified otherwise. 
An attack plan consists of series of exploits along with the relevant vulnerabilities and obtained privileges that would lead the attack to obtain the goal privilege $g\in N_p$.
In the rest of this paper we will use the terms plan and attack path interchangeably referring to the following definition: 
\begin{definition}[Attack path]
\label{def:attack-path}
Attack path $P\subseteq N_p\cup N_e\cup N_c$ is a collection of nodes such that 
\begin{itemize}
\item $g\in P$ -- the goal node is in $P$.
\item $\forall_{p\in P \cap N_p} \exists_{e \in P\cap N_e}, (p,e)\in E$ -- every privilege in $P$ is obtained by executing an exploit.
\item $\forall_{e\in P\cap N_e}\forall_{x : (e,x)\in E}, x\in P$ -- $P$ includes all prerequisites for all its exploits.
\end{itemize}
\end{definition}

In order to execute the attack plan, an adversary should invest effort in executing all the relevant exploits. 
We define the cost of an attack plan as the sum of the normalized costs of exploiting all the relevant vulnerabilities.    
\begin{definition}[Attack cost]
\label{def:cost-path}
The total cost of executing the attack path $P$ is defined as the sum of the exploit costs of all the vulnerabilities included in the path $C_P = \sum_{v\in P\cap N_c} C_v$.
\end{definition}\noindent
We discuss alternative definitions of attack cost in Section~\ref{sec: Conclusions and future work}.

According to Definition~\ref{def:attack-path}, an attack path $P$ may include exploits and vulnerabilities that are not necessary to obtain the goal privilege. 
Assuming that the attacker will avoid unnecessary actions within the target network we define an optimal attack plan as follows.
\begin{definition}\label{def:opt-path}
An optimal attack path in an attack graph $G$ is an attack path having the minimum attack cost:  
$P^*=\argmin_P \{C_P\}$
\end{definition}\noindent
Note that, there might be multiple optimal attack paths in a given attack graph. 
We assume that the attack has the data, knowledge, and capabilities to find the optimal attack plan. 

Next we introduce fake vulnerabilities into the attack graph observed by the attacker in order to manipulate his/her planning decisions. 
Let vulnerability $v\in N_c$ be fake, for example if it is simulated by the Deception Toolkit~\cite{DTK}. The binary function $fake:N_c\rightarrow\{0,1\}$ will be used to discriminate fake configuration nodes from real configuration nodes.    
Any exploit $e\in N_e$ that depends on the fake vulnerability $v$, i.e. $(e,v)\in E$, will fail.
Note that this approach for modeling fake vulnerabilities is different from honeypots used to deceive the attacker~\cite{albanese2016deceiving}, where an exploit will typically succeed but will not provide the attacker with useful privileges.
Please refer to Section~\ref{sec:honeypots} for pros and cons of honeypots based deception. 
%bracha this section exist already???

In this study we assume that the attacker cannot identify fake vulnerabilities before attempting to exploit them. 
An attacker trying to exploit a fake vulnerability $v$ and failing to do so, will realize that this vulnerability is fake and will cease executing the  attack plan in order to avoid unnecessary expenses.  
Given the new observation ($fake(v)=1$) the attacker may discard previous attack plans that rely on fake vulnerabilities and re-plan the attack to find alternative paths.
We assume that all privileges obtained by the attacker due to partial execution of the attack plan remain at his possession and that the attacker will try optimizing the remainder of the attack plan. 
We assume that the attacker will continue doing so until the attack goal is reached.  
Please refer to Section~\ref{sec: Conclusions and future work} for a discussion on alternative attacker models that can be used in this framework. 
We will use the term real attack path to denote a successful attack path that does not rely on fake vulnerabilities to reach the goal: 
\begin{definition}[Real attack path]
\label{def:attack-path}
Let $P\subseteq N_p\cup N_e\cup N_c$ be an attack path according to Definition~\ref{def:attack-path}. $P$ is \emph{real attack path} if and only if $P^`=P-\{c\in N_c : fake(c)=1\}$ is an attack path as well. 
\end{definition}

In this paper we optimize the assignments of fake vulnerabilities such that the overall cost of constructing a real attack path is maximized.  
\begin{definition}[Attack Graph Obfuscation Problem]
Given an attack graph $G=(N_p, N_e,N_c, E, g)$ and a number $k$ of fake vulnerabilities to install find an assignment of fake vulnerabilities $fake:N_c\rightarrow \{0,1\}$ such that:
$\sum_{c\in N_c} fake(c) = k$ and 

\end{definition}

\begin{definition}{PTC(AG)}\label{def:PTC}

PTC(AG) is the minimal, perceived by the attacker,  total cost of the optimal path in attack graph AG (OPT(AG)). Actually, this is the total cost, the attacker estimates it needs to pay, under the assumption that all the appeared vulnerabilities in the attack graph are valid, and not fake.\\ 
In other words, $PTC(AG) = C_{s,t}(OPT(AG))$ where s and t are the first and last nodes in OPT(AG) respectively.\\
In practice, we compute PTC(AG) as the cost of the output of a planner which produces an optimal attack plan, given attack graph AG. 
\end{definition}

\begin{definition}{$AA(AG,\left\{a_1,...,a_n\right\}$}\label{def:AA}
The output is an attack graph AG', constructed from assigning the group of assignments: $\left\{a_1,...,a_n\right\}$.
\end{definition}
\begin{definition}{APTC(AG)}\label{def:APTC}
The actual cost the attacker will pay when trying to reach its target in attack graph AG.
Calculating APTC(AG) is as follows (also described in Algorithm \ref{algo: APTC}):
\newline\noindent\textbf{step 1: } We initialize the aptc variable (that accumulates the cost) to zero. % and denote AA(AG,assignments) as AG'. \\
\newline\noindent\textbf{step 2: } If the optimal attack path in AG consists of vertexes that were created due to fake vulnerability assignments (fake vertexes), we add to aptc the cost from the initial vertex to the first fake vertex in the attack plan.
\newline\noindent\textbf{step 2.1} Now we modify the attack graph AG as follows:
\begin{enumerate}
\item We remove the fake assignment, which removes all the vertexes created due to that.
\item We set to zero the weights of the vertexes the attacker visited on the path to the first fake vertex. 
\end{enumerate}
Then, we iterate to step 2 with the new attack graph AG.\\
The algorithm stops when the optimal attack path from the source to the target has no more fake vulnerabilities. The algorithm returns the accumulated aptc. 

\begin{algorithm}[ht]
\caption{APTC(AG)}
\label{algo: APTC}
\begin{algorithmic}[1]
\renewcommand{\algorithmicrequire}{\textbf{Input:}}
\renewcommand{\algorithmicensure}{\textbf{Output:}}
\Require AG
\Ensure APTC(AG)
\State aptc $\gets$ 0
\State $i \gets 1$
\State Let $Assignments = ((IP_1,v_1), (IP_2, v_2) ... (IP_m,v_n))  \Leftarrow  $ all the fake assignments of fake vulnerabilities to IPs.
\State AG' $\gets$ AG
\State $p_i=(u_1, u_2, ... u_m) \gets OPT(AG')$
\State flag = True
\While{flag}
\State flag = False
\If{$ \exists k$, such that $u_k$ is a vertex created from some fake assignment $a \in Assignments$}
\State Let $u_k$ be the first node which maintains the above.
\State flag = True
\State aptc = aptc + $\sum_{j=1}^{k} C(u_j)$
\For {vertex $\in$ OPT(AG')}
\If {vertex != $u_k$}
\State cost($u_k$) = 0
\Else
\State break
\EndIf
\EndFor
\State AG' $\gets$ removeNodeFromAttackGraph(AG',a)
\State $p_i=(u_1, u_2, ... u_m) \gets OPT(AG')$
\EndIf
\State $i \in i+1$
\EndWhile
\end{algorithmic}
\end{algorithm}
\end{definition}
\begin{lemma}\label{lemma:PTC(AG)=APTC(AG)}
If attack graph AG has no fake vulnerabilities, then \\PTC(AG) = APTC(AG).
\end{lemma}
\begin{proof}
If attack graph AG has no fake vulnerabilities, there are no re-calculation of attack plans. 
Due to the fact that the attack plan produced from attack graph AG, it does not include vertexes that were added due to fake vulnerabilities addition. From definition \ref{def:APTC}, steps 1 and 2 does not accrue at all. So actually APTC(AG) is the total cost of the attack plan produced to AG. In other words -  APTC(AG) = PTC(AG).
\end{proof}

%Random based approach
\subsubsection{Random based approach} \label{subsub: Random based approach}
For the random baseline approach we obfuscate the attack graph by choosing the deceptive IPs and the vulnerabilities randomly. \\
In order to achieve the condition state in equation (1) above, we generate the following algorithm (see also algorithm \ref{algo: create AG'} and Figure \ref{fig:CreateG'}):\\
%%%%%%%%%%%%%%%% begin Figure %%%%%%%%%%%%%%%%%%%
\begin{figure}
\centering
\includegraphics[width=0.7\linewidth]{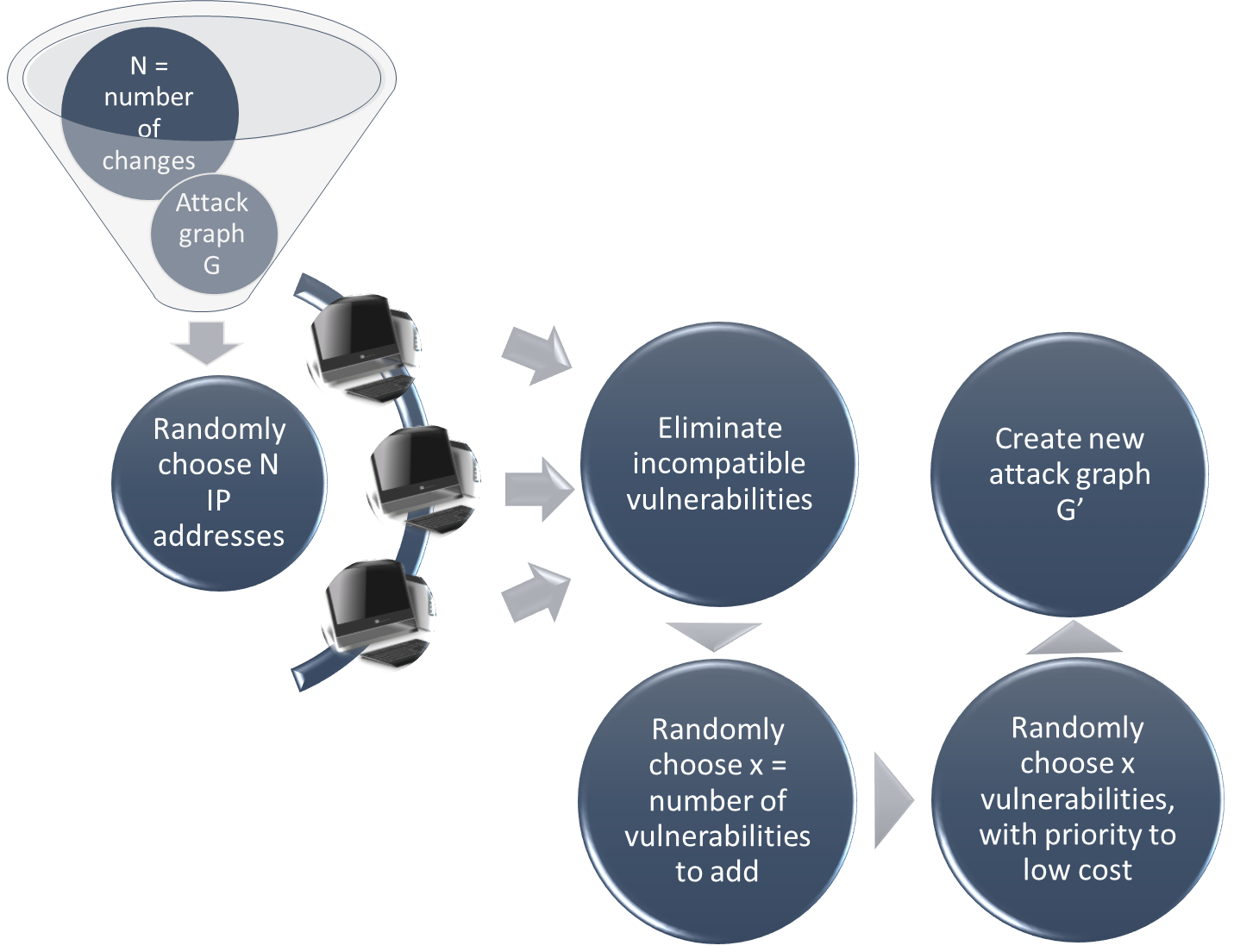}
\caption{Random approach: create AG'}
\label{fig:CreateG'}
\end{figure}
%%%%%%%%%%%%%%%% end Figure %%%%%%%%%%%%%%%%%%%
\begin{algorithm}[ht] 
\caption{Create AG'}
\label{algo: create AG'}
\begin{algorithmic}[1]
\renewcommand{\algorithmicrequire}{\textbf{Input:}}
\renewcommand{\algorithmicensure}{\textbf{Output:}}
\Require AttackGraph AG, numofChanges
\Ensure AG'
\State AG' = AG
\State $pickedIps \leftarrow$ list of random chosen IPs from $AttackGraph$
\For {each IP $ip\in pickedIps$}
\State $currentVulList \leftarrow$ list of all vulnerabilities exist in $ip$
\State $OSVersion \leftarrow$ getOSVersion$(RV_u)$
\State $vulList \leftarrow$ list of all vulnerabilities associate with $OSVersion$ 
\State $validOSList \leftarrow$ $vulList \setminus currentVulList$
\State $numOfVulToAdd \leftarrow Random(0,validOSList\leftarrow Size())$
\State $chosenVul \leftarrow $ \\$chooseRandomList (validOSList,numOfVulToAdd)$
\State $AddNodes(AG', IP, chosenVul)$
\EndFor
\Return AG'
\end{algorithmic}
\end{algorithm}

Given an attack graph AG and N number of desired changes, \\
\textbf{step 1:} We randomly choose x \big[=(Number of IPs in original network) $\times$ (Portion of desired changes in the graph)\big] IPs that will be designated as deceptive hosts. \\
 \textbf{step 2:} We eliminate incompatible vulnerabilities, by filtering vulnerabilities that create conflicts with other information known about the given IP address. This is done by filtering vulnerabilities that do not match the operating system that exists on the targeted computer. The information gathered about the IP address was collected earlier by the Nessus scanner \cite{nessus}.
\\ \textbf{step 3:} We randomly choose how many vulnerabilities to add to each IP address. This is done because we want to be as unpredictable as possible, in order to be invisible and avoid detection by the attacker.\\
In \textbf{step 4:} we choose %bracha - I think this should not be x as we another x on the first step. May be V??
x vulnerabilities for each IP address chosen as deceptive. These vulnerabilities are randomly selected, with preference to low-cost vulnerabilities. This is done based on the assumption that adding the same vulnerabilities can lead to discovery of the defender's deception by the attacker.\\
\textbf{step 5:} We create new nodes in the attack graph. For every chosen IP address, we add the new chosen vulnerabilities to the Nessus file. Then, we regenerate the attack graph with MulVAL, using the same connectivity as in the original graph.
As a result, new paths to the target are created, some of them consists  only of fake vulnerabilities, and some are combined from fake and real vulnerabilities.

%AI search based approach
\subsubsection{AI search based approach}\label{subsub: AI search based approach}
In this approach we obfuscate the  attack graph AG into attack graph AG', by using  heuristics.\\
Given an attack graph AG and a number of changes K, we need to find the best K assignments of fake vulnerabilities in the given network. The best K assignments are the K assignments that will maximize the adversary's total attack cost. Thus, we need to find the AG' ,that is generated by  K fake vulnerabilities assignments to AG, which maximize APTC(AG').\\
We model the problem as a search problem and apply the heuristics described in \cite{puzis2007finding}, adapting them into our problem, and add new heuristic and elements ordering.

First, we model the problem of vulnerabilities assignments into a search problem, as  we search  the group of best k assignments, by adapting the search space described in \cite{puzis2007finding}:
%Search space
\paragraph{Search space} \label{search space}\mbox{}\\
The search space is constructed as a decision tree.
In the following discussions we will use the terms nodes and transitions to refer the graph elements of the \textit{search space} and the terms vertex and edges to refer to the graph elements of the \textit{input attack graph}.
Every node (denoted by V) maintains two collections:
\begin{enumerate}
\item Assignments (denoted by A(V)) - group of vulnerability assignments to an IP address. for example: \\ $<Vulnerability_{1}, IP_{1}>, <Vulnerability_{2}, IP_{2}>,$ ... $<Vulnerability_{l}, IP_{l}>$
\item Candidate assignments
(denoted by AC(V)) - an ordered list of possible vulnerability assignments to an IP address. for example: \\ 
$<Vulnerability_{1}, IP_{1}>, <Vulnerability_{2}, IP_{2}>,$ ... $<Vulnerability_{l}, IP_{m}>$
\end{enumerate}
The order of assignments candidate in AC(V) may vary between different assignments in A(V) group. 
The methods of ordering the assignments in AC(V) will be described in Section \ref{par: Element ordering} we will refer to the first assignment in the ordered list AC(V) as ${ac}^{best}$.
We will refer to the search space as a decision tree. Thus, at every node we decide whether ${ac}^{best}$ is included in the final k assignments.
According to
this decision we branch to either of the two sub-trees, one containing all the groups that include ${ac}^{best}$ and the other containing all the groups that do not include it. Thus we assure that all possible solutions can be found during the search.\\
Every node has two children: 
\begin{itemize}
\item Left child (denoted by $V^-$).
\begin{itemize}
\item Has assignments group as $V$:\\ 
A($V^-$) = A(V)
\item Has candidate list which contain all
the assignments from the candidates list of $V$
excluding ${ac}^{best}$:\\
AC($V^-$) = AC(V) $\setminus$ $\left\{{ac}^{best}\right\}$
\end{itemize}

\item Right child (denoted by $V^+$).
\begin{itemize}
\item Has all the assignments of $V$,
additionally including ${ac}^{best}$:\\ 
A($V^+$) = A(V) $\cup$ ${ac}^{best}$
\item Has candidates list which contain all
the assignments from the candidates list of $V$
excluding ${ac}^{best}$:\\
AC($V^+$) = AC(V) $\setminus$ $\left\{{ac}^{best}\right\}$
\end{itemize}
\end{itemize}

The root node of the tree represents an empty group of assignments. Its candidates list contains all the possible assignments to the graph, under the assumption that if the cost of vulnerability x and vulnerability y equals and they have the same pre and post conditions, assignment of vulnerability x to a specific IP is the same as assigning y to the same IP.
Every child of the root node holds a candidate list that is missing one assignment and a group that includes one assignment or none.
All nodes at the second level hold a candidates list that is missing two assignments and a group that includes zero, one, or two assignments, etc.\\
Since we are looking for a group of size K, we will refer to nodes with $\left|A(V)\right|=K$ as possible solutions to our problem. We exclude from the search space all sub-trees that do not contain a possible solution. These can be either sub-trees rooted at nodes with $\left|A(V)\right|>K$ or sub-trees rooted at nodes with $\left|AC(V)\right| + \left|A(V)\right| <K$. \\
%%%%%%%%%%%%%%%% begin Figure %%%%%%%%%%%%%%%%%%%
\begin{figure}[ht]
\centering
\centering
\includegraphics[width=.7\linewidth]{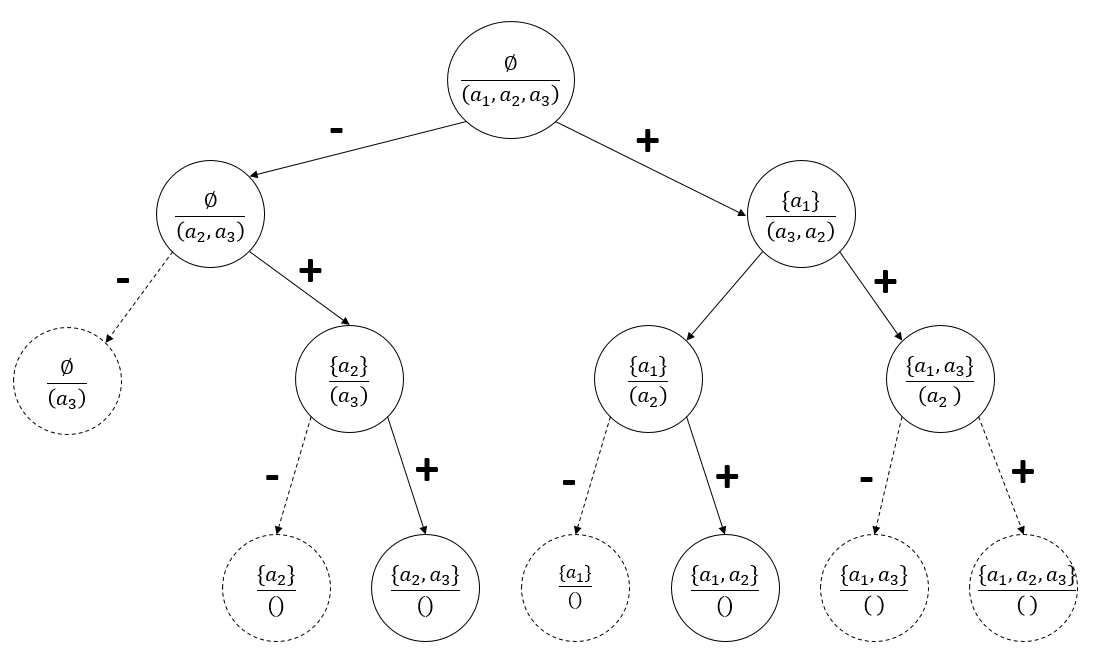}
\caption{Example of a tree representing the search space}
\label{fig:searchSpace}
\end{figure}
%%%%%%%%%%%%%%%% end Figure %%%%%%%%%%%%%%%%%%%
Figure \ref{fig:searchSpace} presents an example of a tree representing the search space of the vulnerabilities assignment problem with K=2 and l=3. Dashed lines represent sub-trees that do not contain a possible solution, and therefore will be pruned. The set in curly braces $\left\{...\right\}$ represents A(V) and the list in parenthesis (...) represents AC(V). 
\par The search should start at the root node R whose assignment group is empty and the candidates list is $(a_1, a_2, a_3)$ and ${ac}^{best}$ = $a_1$. We consider all groups where $a_1$ is a member by moving to the right child (following the "+" transition).\\
$R^+$ maintains the group $\left\{a_1\right\}$ and a candidates list that includes all possible assignments, excluding the assignment $a_1$. It is possible that the order of the assignments in $AC(R^+)$ is different from the order of assignments in $AC(R)$ (See section \ref{par: Element ordering}). We consider all groups where $a_1$ is not a member by following the transition "-" from the root node. $R^-$ maintains an empty group and candidates list that is identical to the candidates list of $R^+$.

\paragraph{Heuristics Search Algorithm}\mbox{}\\
The input to this algorithm is:
\begin{enumerate}
\item AG - the attack graph produced from the given network.
\item P - attack planner.
\end{enumerate} 
We investigated the effect of two search algorithms on this problem - DF-BnB and A*.
\par \underline{\textbf{A*}}:\\ 
We use the A* algorithm to search for a quicker solution to our problem. A* is a best-first search algorithm, meaning that it solves problems by searching among all possible paths to the solution while considering first the paths that appear to lead most quickly to the solution. \\\\

\par \underline{\textbf{DF-BnB}}:\\ 
We use the Depth First Branch and Bound (DF-BnB) algorithm \cite{zhang1995performance} to search the tree for the group of assignments which produce an attack plan with maximal attacker's cost.\\ 
DFBnB is known to be effective when the depth of the search tree is known, as in our case. DFBnB is similar to DFS but it prunes nodes according to a global bound. Under the assumption that the given planner (P) finds the optimal solution, we begin the search with a bound equal to PTC($AG$). Which is the cost of the optimal path in attack graph AG. During the search, the bound is equal to the maximal utility found so far.\\\\
\par In contrast to the traditional search that is aimed at finding the node with minimal cost, we are using a utility based approach. We define the utility of the node V (denoted by U(V)) as APTC(AA(AG,A(V))).\\
A heuristic function h(V) estimates the maximal utility that can be gained by exploring the sub-tree rooted at V. f(V) = g(V)+h(V) is a function that estimates the maximal utility of nodes in V's sub-tree.\\
The utility of the root node R is equal to $APTC(AG,\left\{\right\})=PTC(AG)$ (see lemma \ref{lemma:PTC(AG)=APTC(AG)}), because A(R) = $\emptyset$ and adding zero assignments to the attack graph produce an attack plan with optimal cost. h(R) is an upper bound on the optimal solution. While searching down the tree, g(V) will grow and h(V) will decrease. When the
algorithm finds a possible solution, h(V) is equal to zero and f(V) is equal to APTC(AA(AG,A(V))). While necessarily $\left|A(V)\right|=K$.\\
Pruning decisions made during the search are based on the value of f(V) and the maximal APTC found so far. We can guarantee that the heuristic search will find the optimal solution only if the function f(V) is an upper bound on the maximal APTC that can be found within the sub-tree rooted at V. If this upper bound is below the maximal
APTC found so far, the sub-tree is pruned, otherwise it is explored in hope of finding a group with a higher APTC. Heuristic functions used for pruning nodes in the search tree are described in section \ref{par: Utility heuristic}.\\
When visiting a node, beside the decision whether to prune the current sub-tree, the algorithm should also determine ${ac}^{best}$. The following section describes two methods of ordering assignments in AC(V) and determining ${ac}^{best}$.
%Element ordering
\paragraph{Assignments ordering in the candidates list}\label{par: Element ordering}\mbox{}\\
Admissible heuristic functions guarantee that during the search we will find the optimal solution (note that the search may take an exponential time). The order of the assignments in AC(V) is important for choosing ${ac}^{best}$ and for computing admissible pruning heuristics. The order of assignments in AC(V) can be determined by either their individual utility or their contribution to the given chosen assignments:\\
\subparagraph{High utility first ordering}\mbox{}\\
We can sort the assignments in AC(V) by their individual utility (the utility of an assignment $a \in A(V)$ is APTC(AA(AG,{a}))). That means that the first assignment in A(V) will be the assignment with the highest individual utility.
The computation of APTC(AA(AG,{a})) needs to be done in the pre-processing phase, for each assignment. But, it can be done once for the entire search. Sorting assignments by their individual utility will impose the same order on all candidates lists in the tree. \\
The ideal placement of a fake vulnerability will be by making each assignment create a new path to the goal with a lower cost than the optimal path in the given attack graph AG ($PTC(AA(AG,\left\{a_1\right\}))<PTC(AG)$, when $a_1\in A(V)$). If this situation is impossible, we will need couple of assignments in order to create new path. But, after the attacker will reveal the first fake vulnerability, the other assignments that created the new path, are now redundant. Dur to that, this solution is inferior to the previews one, due to it's wastefulness.\\
By ordering the assignments list by the assignment with the highest utility first, we can find paths that created due to one fake vulnerability assignment.

\subparagraph{Shortest path first ordering}\mbox{}\\
We can sort the assignments in AC(V) by considering the assignments selected so far.\\
If the solution described in the previous paragraph is uncommon, we need to find the shortest paths we can produce to the goal. That means, we need to find the smallest group of fake assignments that can give us a new path to the goal with a lower cost than the optimal path in the given attack graph AG. Thereby we are using our budget carefully.\\
Due to that, the contribution of the assignments in CA(V) should be recalculated each time the search algorithm moves to the right child (by adding ${ac}^{best}$ to A(V)). In pre-processing phase we calculate all the possible paths from the source to the target in the given attack graph AG, and ordering them from the shortest path to the longest. According to the assignments in A(V), we need to find the smallest group of Ips that can create a path to the goal. We used the inverted index data structure in order to that efficiently.\\
In Figure \ref{fig:invertedIndex} we show an example of extracting the best candidates when given a node V in the search tree. By finding for every $IP\in a_i$, where $a_i\in A(V)$, the paths it is part of. Then finding the smallest group of best candidates.
%%%%%%%%%%%%%%%% begin Figure %%%%%%%%%%%%%%%%%%%
\begin{figure}[ht]
\centering
\includegraphics[width=1\linewidth]{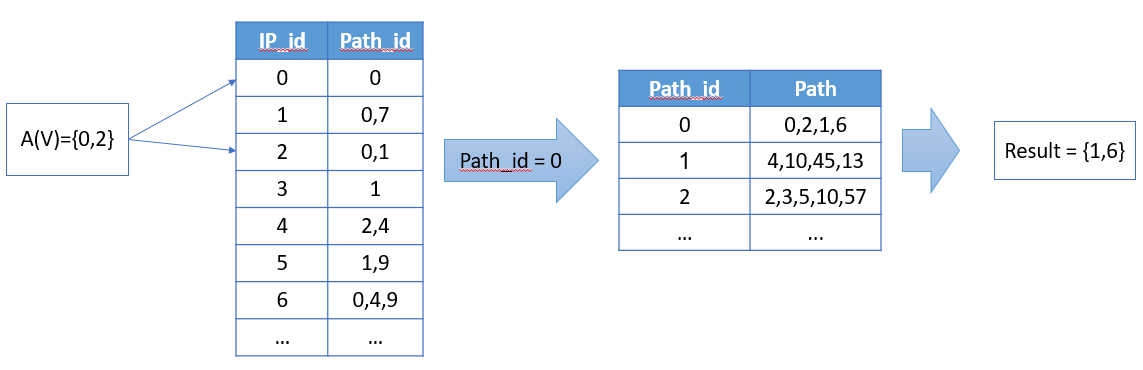}
\caption{Example of extracting the best IP addresses candidates using inverted index}
\label{fig:invertedIndex}
\end{figure}
%%%%%%%%%%%%%%%% end Figure %%%%%%%%%%%%%%%%%%%

\paragraph{Heuristic evaluation functions} \label{par: Utility heuristic}\mbox{}\\
In this problem, we need to find the assignment which produces the maximum utility, it is a maximum problem.  A heuristic function in a maximum problem is admissible, if in every node v, in the search space, it holds that h(v), which  is an upper bound on the remaining reward of an optimal solution from the start to a goal via v (denoted as $h^*(v)$) \cite{stern2014max}.\\
The proposed functions differ in their computation time and precision.
\subparagraph{$h_1$ (utility upper bound)}
$$h_1(V) = \sum_{i=1}^{K-\left|A(V)\right|}(APTC(AA(AG,\left\{{ac}_{i}\right\}))),$$ 
\centerline{Where $ac_i\in$ AC(V) is the i'th place in the candidate list AC(V).}\\[5pt]
In $h_1(V)$, we sum the APTC of the graph constructed from applying each assignment from the K - $\left|A(V)\right|$ assignments                                                                                                                                           in the ordered candidates list AC(V). \\[5pt]
%%%%%%%%%%%% time complexity
\textbf{\underline{Admissibility:}}\\
$h_1$ is not admissible. We will show an example in order to show that:\\
%%%%%%%%%%%%%%%% begin Figure %%%%%%%%%%%%%%%%%%%
\begin{figure}[ht]
\centering
\centering
\includegraphics[width=1\linewidth]{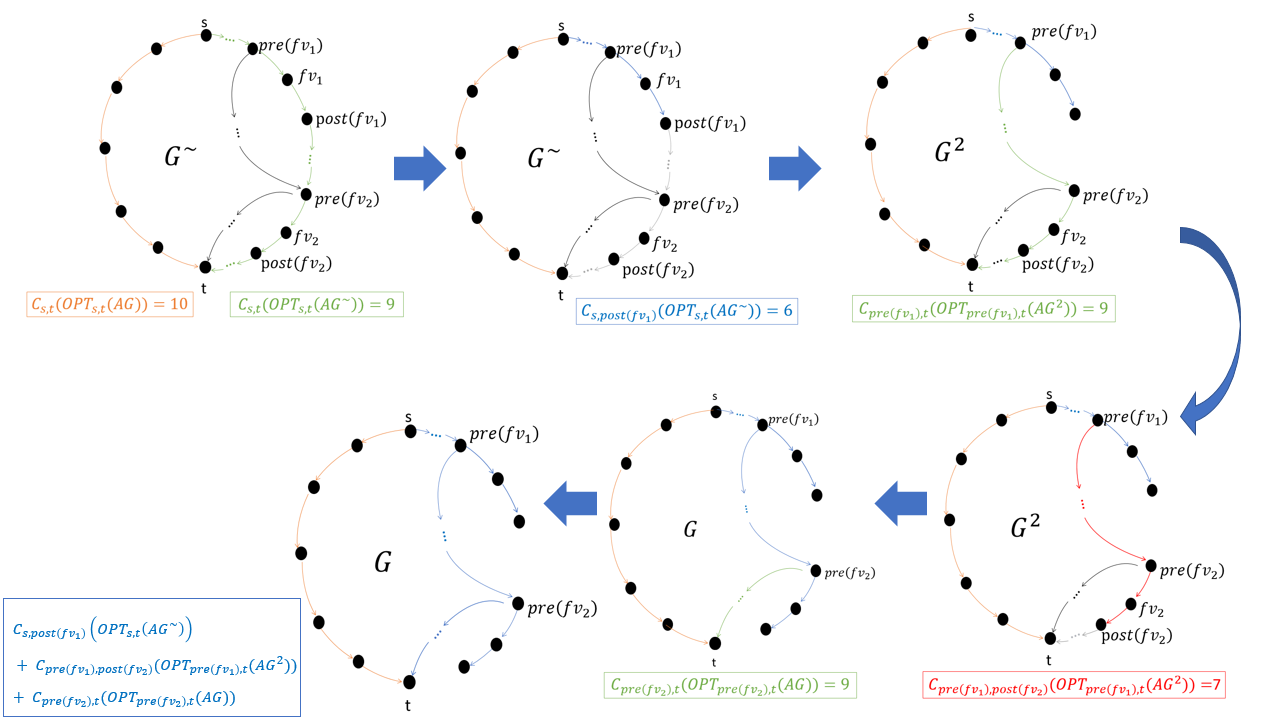}
\caption{Example of non-admissibility of heuristic $h_1$}
\label{fig:admisible_all}
\end{figure}
%%%%%%%%%%%%%%%% end Figure %%%%%%%%%%%%%%%%%%%
% Let $C_{i,j}(AG)$ be the optimal path from vertex i to vertex j in attack graph AG.\\
We denote $pre_p(i)$ as the previous vertex of vertex i in path p.\\
We denote $post_p(i)$ as the next vertex of vertex i in path p.\\
%We denote $C_G(v_1,v_2)$ as the cost of the optimal path from vertex $v_1$ to vertex $v_2$. In other words, it is the minimum cost it take to go from vertex $v_1$ to vertex $v_2$ in graph AG. \\
Let AG be the given original graph.
Let the root node in the search tree be R, where A(R) = $\emptyset$ and AC(R) = {$c_1, c_2$}, where $c_1 = (IP_1,v_1)$ and $c_2 = (IP_2,v_2)$ and $IP_1 \ne IP_2$.\\
We denote the attack graph,  produced from assigning the fake assignments $c_1$ and $c_2$, as $AG^\sim$. $AG^\sim = AA(AG,\left\{c_1, c_2\right\})$.\\
Assume that applying both $c_1$ and $c_2$ creates a new path to the target vertex (denoted as $p^\sim$) and PTC($AG^\sim$)<PTC(AG).\\ 
Furthermore, let's assume that applying only $c_1$ or $c_2$ doesn't create a new path to the goal such that APTC($AG^1$)=APTC(AG) and APTC($AG^2$)=APTC(AG), when $AG^1$ and $AG^2$ are the graph constructed from adding assignment $c_1$ or assignment $c_2$ to attack graph AG, accordingly.\\
Let's assume that the new path $p^\sim$ is as follows: $p^\sim = (s,...,(fv)_1,...,(fv)_2,...,t).$ Where $(fv)_1$ and $(fv)_2$ be the fact vertexes added due to adding $c_1$ and $c_2$ assignments accordingly.\\
For simplicity, let's assign the following costs:\\
\centerline{$C_{s,t}(OPT_{s,t}(AG))=10$,   $C_{s,t}(OPT_{s,t}(AG^\sim))=9$,   $C_{s,post(fv_1)}(OPT_{s,t}(AG^\sim))=6$, } \centerline{ $C_{pre(fv_1),t}(OPT_{pre(fv_1),t}(AG^2))=9$,   $C_{pre(fv_1),post(fv_2)}(OPT_{pre(fv_1),t}(AG^2))=7$,}   \centerline{$C_{pre(fv_2),t}(OPT_{pre(fv_2),t}(AG))=9$ }\\
We describe the case of applying both $c_1$ and $c_2$:\\
\\
First, the adversary's attack plan is equal to path $p^\sim$ (because\\ $PTC(AG^\sim) < PTC(AG)$). When it tries to preform the attack plan, it starts at vertex s, and continue until it tries to exploit $post_{p^\sim}(fv_1)$ - the action vertex enabled due to applying vulnerability $v_1$ to IP $IP_1$. The attacker will fail when trying to exploit this vertex, due to the fact that the vulnerability is fake. The adversary now removes the fake vulnerability $v_1$ from IP $IP_1$, and constructs a new attack graph which is similar to $AG^2$, but does not loose the privileges he already gained. Thus, The attacker will try to find the best path from s or $pre(fv_1)$ to t. Since $C_{pre(fv_1),t}(OPT_{pre(fv_1),t}(AG^2))<C_{s,t}(OPT_{s,t}(AG))$ the attacker will choose the optimal path from $pre(fv_1)$ to t.\\
It will continue to exploit vulnerabilities until it reaches $fv_2$ – tries to exploit it and fails, again. The adversary now constructs AG again, while not loosing all the privileges it gained so far. Now, $C_{pre(fv_2),t}(OPT_{pre(fv_2),t}(AG))<C_{s,t}(OPT_{s,t}(AG))$ so the adversary will choose the path from $pre(fv_2)$ to t.\\
So, to sum up the total cost for the adversary during its path:\\
$APTC(AG^\sim)= C_{s,post(fv_1)}(OPT_{s,t}(AG^\sim)) + 
C_{pre(fv_1),post(fv_2)}(OPT_{pre(fv_1),t}(AG^2)) + 
C_{pre(fv_2),t}(OPT_{pre(fv_2),t}(AG)) =6+7+9=22 > 9 = PTC(AG^\sim)$ \\
But, $h_1(R) = APTC(AA(AG,\left\{{c}_{1}\right\}))+APTC(AA(AG,\left\{{c}_{2}\right\}))= 20 < 22 = APTC(AA(AG,\left\{{c}_{1},{c}_{2}\right\}))=h^*(R)$

\subparagraph{$h_2$ (admissible utility upper bound)}\label{sub:h2 admissible}\mbox{}\\ 
$h_2$ is similar to $h_1$, but with an addition:
$$h_2(V) = PTC(AG) + \sum_{i=1}^{K-\left|A(V)\right|} (APTC(AA(AG,\left\{{ac}_{i}\right\}))),$$ 
\centerline{Where $ac_i\in$ AC(V) is the i'th place in the candidate list AC(V).}\\[5pt]
In $h_2(V)$, we sum PTC(AG) to the APTC of the graph constructed from applying each assignment from the K - $\left|A(V)\right|$ assignments in the ordered candidates list AC(V)\\[5pt]
\textbf{\underline{Admissibility:}}\\
The heuristic $h_2$ is admissible, in order to prove it, we need to prove that for all $v$ the following is true:\\
$$h_2(v) \geq h^*(v)$$
We used the following lemma in order to prove it:
\begin{lemma}\label{lemma:k+1 bounded}
Let AG be an attack graph with no fake vulnerabilities, then:\\ $ATPC(AA(AG,\left\{a_1,a_2,...,a_k\right\}))\leq (k+1)*ATPC(AG)$
\end{lemma}
\begin{proof}
We will prove it using induction:\\
\textbf{Basis:} k=0\\
$ATPC(AA(AG,\left\{\right\})) = APTC(AG) \leq (0+1)*ATPC(AG)$\\
Let $AG^(k-1) = AA(AG,\left\{a_1,a_2,...,a_(k-1)\right\})$
Let's assume that for k-1 the following is true:\\
$ATPC(AG^(k-1))\leq k*ATPC(AG)$\\
We add another assignment to the attack graph $AG^(k-1)$, so $AG^(k) = AA(AG^(k-1),\left\{a_k\right\})$. The best case is that the new assignment constructs a new path to the goal such that $C_(s,t)(OPT_(s,t)(AG^(k))\leq C_(s,t)(OPT_(s,t)(AG^(k-1))$ and of course $C_(s,t)(OPT_(s,t)(AG^(k)) \leq C_(s,t)(OPT_(s,t)(AG)$. The attacker will try to execute $OPT_(s,t)(AG^(k))$, then after trying to exploit the fake vulnerability, it fails and tries to find a new path in $AG^(k-1)$ (let's assume that in the best case, the privileges gained so far did not help). So, in the best case: $ATPC(AG^k) = C_(s,t)(OPT_(s,t)(AG) + ATPC(AG^(k-1)) \leq PTC(AG) + k*ATPC(AG) =^* (k+1) * ATPC(AG) $ \\
* - lemma \ref{lemma:PTC(AG)=APTC(AG)}
\end{proof}
\begin{lemma}\label{lemma:APTC smaller then}
Let AG be an attack graph with no fake vulnerabilities, and $a_i=<IP_i, Vul_i>$ is a possible assignment, then:\\ $ATPC(AA(AG,\left\{a_i\right\}))\geq ATPC(AG)$
\end{lemma}
\begin{proof}
If the assignment of $Vul_i$ in $IP_i$, creates a new optimal path to the goal (that means $PTC(AA(AG,\left\{a_i\right\}))\leq PTC(AG)$), so $ATPC(AA(AG,\left\{a_i\right\}))$ will be equal to the sum of the following paths:
\begin{itemize}
\item First, the attacker will construct an attack plan from the attack graph $AG' = AA(AG,\left\{a_i\right\})$. The adversary will try to exploit all the exploit vertexes to the goal, until it will try to exploit  $Vul_i$ and fails. Until then, it will pay \textbf{$C_{s,action_(Vul_i)}(AG')$} - the cost of the path from source (s) to the action vertex added due to the existence of the fake vulnerability $Vul_i$. 
\item Now, after the adversary realizes that the vulnerability it tries to exploit ($Vul_i$) is fake, it will remove the fake information from the attack graph and re-plans its path to the goal. The new attack graph considers the fact that the attacker does not loose privileges. So x=$min \left\{ C_{s,t}(OPT(AG)),C_{prev(action_(Vul_i)),t}(OPT(AG)) \right\} $ is added to path cost of the previous step.
\end{itemize}
%bracha - the following senetnce is not clear.. something is wrong with the sentence 
%Hadar - I add more explanations, I hope it is clear now
If $ATPC(AA(AG,\left\{a_i\right\})) = C_{s,action_(Vul_i)}(AG') + C_{s,t}(OPT(AG))$ then\\ $ATPC(AA(AG,\left\{a_i\right\}))\geq C_{s,t}(OPT(AG)) = ATPC(AG)$ \\\\
If $ATPC(AA(AG,\left\{a_i\right\})) = C_{s,action_(Vul_i)}(AG') +C_{prev(action_(Vul_i)),t}(OPT(AG))$ then\\ 
$ATPC(AA(AG,\left\{a_i\right\}))\geq C_{s,t}(OPT(AG)) = ATPC(AG)$ \\
otherwise, let $p'= <s, ..., prev(action_(Vul_i), ..., t >$ the path from the source s to the target t via the previous vertex to $action_(Vul_i)$, then $C_{s,t}(p') < C_{s,t}(OPT(AG)) $ in contrast to the Optimality of OPT(AG). So, necessarily $ATPC(AA(AG,\left\{a_i\right\}))\geq ATPC(AG)$.

%If we assume that  $ATPC(AA(AG,\left\{a_i\right\})) < ATPC(AG)$ so the cost of the path $p' = <s, ..., prev(action_(Vul_i), t >$, $C_{s,t}(p')< C_{s,t}(OPT(AG)) $ in contrast to that OPT(AG) is optimal in AG.

\end{proof}\mbox{}\\
Using the above lemma we prove the \textbf{admissibility of $h_2$}.
\\We can see that $h^*(V) = APTC(AA(AG,\left\{ac_1,ac_2,...,ac_{K-\left|A(V)\right|}\right\})) - \\APTC(AA(AG,A(V)))$ \\
Where $ac_i$ is the i'th element in $AC(V)$.\\
Because it is the remaining reward of an optimal solution from the start to a goal via V\\[5pt]
$h_2(V) = PTC(AG) + \sum_{i=1}^{K-\left|A(V)\right|} (APTC(AA(AG,\left\{{ac}_{i}\right\}))) \geq^{Lemma \ref{lemma:APTC smaller then}} PTC(AG)+ (K-\left|A(V)\right|)*APTC(AG) \geq^{Lemma\ref{lemma:PTC(AG)=APTC(AG)}} APTC(AG)+ (K-\left|A(V)\right|)*APTC(AG) = (K+1-\left|A(V)\right|)*APTC(AG)\geq^{Lemma \ref{lemma:k+1 bounded}} APTC(AA(AG,\left\{a_1,a_2,...,a_{K-\left|A(V)\right|}\right\})) \geq h^*(V)$
\\We can see that $h_2(V)\geq h^*(V)$ and as result $h_2$ is admissible.
%We will prove the above using induction on the number of summed candidates ($K-\left|A(V)\right|$), given an attack graph AG:\\
% \textbf{\underline{Basis:}} $K-\left|A(V)\right|$ = 0:\\
% We got a valid solution, so there are no valid candidates in A(V), so:\\
% $$h_2(v) = PTC(AG) + \sum_{i=1}^{0}(APTC(AA(AG,{ac_i}))) = PTC(AG) =^* APTC(AG) > 0 = h^*(v)$$
% Where $ac_i\in$ AC(V) is the i'th place in the candidate list AC(V).\\
% * - AG is the given graph, so it does no contain fake vulnerabilities. See lemma \ref{lemma:PTC(AG)=APTC(AG)}\\ 
% \textbf{\underline{Inductive step:}} Let's assume that for n:\\ 
% $PTC(AG) + \sum_{i=1}^{n}(APTC(AA(AG,{ac_i}))) \geq APTC(AA(AG,(A(v_n)))$  holds.\\
% We need to prove that the following holds for n+1:\\
% $PTC(AG) + \sum_{i=1}^{n+1}(APTC(AA(AG,{ac_i}))) \geq APTC(AA(AG,(A(v_n)))$ \\
% We will divide the proof into two cases:

%Cost heuristic
\paragraph{Assignments cost} -  under the constrain that we have K assignments to add to the given attack graph, the cost of every assignment $a\in A(V)$ is 1. Hence, if the budget is x, we can choose x assignments to add, and that is exactly what we want.

\paragraph{Pre-Processing}\mbox{}\\
During the pre-processing step, we calculate the $PTC(AG,\left\{a\right\})$, for each potential $a \in AC(root)$.\\
Further more, we eliminate incompatible vulnerabilities for each IP, by filtering vulnerabilities that create conflicts with other information known about the given IP address. This is done by filtering vulnerabilities that do not match the operating system that exists on the targeted computer. 
\\
To enable the filtering, information about the IP address is collected using the Nessus scanner \cite{nessus}. The information about each vulnerability is gathered from the National Vulnerability Database \cite{NVD}.
%Utility heuristic

%%%%%%%%%%%%%%%%%%%%%%%%%%%%%%%%%%%%%%%%%%%%%%%%%%%%%%%%%%%%%%%%%%
%% Evaluation
%%%%%%%%%%%%%%%%%%%%%%%%%%%%%%%%%%%%%%%%%%%%%%%%%%%%%%%%%%%%%%%%%%
\section{Evaluation} 
\label{sec:eval}
We need to produce an attack graph in order to  model all the possible paths from the adversary's source (usually the Internet) to its target in the network.
Two types of input are required
to create an attack graph: a list of existing vulnerabilities in the network hosts (Data collection), which can be gathered using Nessus scan \cite{nessus}, and the topology of the network. 

\subsection{Data} 
\label{sec:data}
In order to measure the effectiveness of our method we needed real data that represent a real network. 
We collected data using a Nessus \cite{nessus} scan of an \textbf{organization}. 
This network consists of 150 hosts and had 394 vulnerabilities. 
We used an automatic tool, called MulVAL, for constructing the attack graph from the real network.
\note{RP}{Cite our previous papers for the data.}

The organizational network exhibited a standard topology of a small enterprise: 
\begin{itemize}
\item DMZ - the network accessible from the Internet.
\item Internal network.
\item Secured network - the network in which all the important assets were located.
\end{itemize}

\subsection{Attack graph generation tool}\label{subsub: Attack graph generation tool}
Due to all the strengths we mention in section \ref{sec:lag}, we chose logical attack graph as a model for our problem. In order to produce a logical attack graph, we used a state of the art attack graph generation tool called MulVAL \cite{ou2005mulval}.

\paragraph{Attack planner} \label{subsub: Attack planner}
In order to obtain the optimal attack plan, with minimal cost, we use Jorg Huffman's converter \cite{sarraute2013penetration} \cite{sarraute2013pomdps} in order to convert our problem to PDDL (Planning Domain Definition Language) file.\\
We solve the PDDL file obtained in the previous step using the fast-downward planner \cite{helmert2006fast} which we executed with the A* search algorithm using the landmark-cut heuristic function. The landmark-cut function is admissible, so the algorithm would eventually find an optimal solution to the goal.

%%%%%%%%%%%%%%%% begin Figure %%%%%%%%%%%%%%%%%%%
\begin{figure}[ht]
\centering
\centering
\includegraphics[width=1\linewidth]{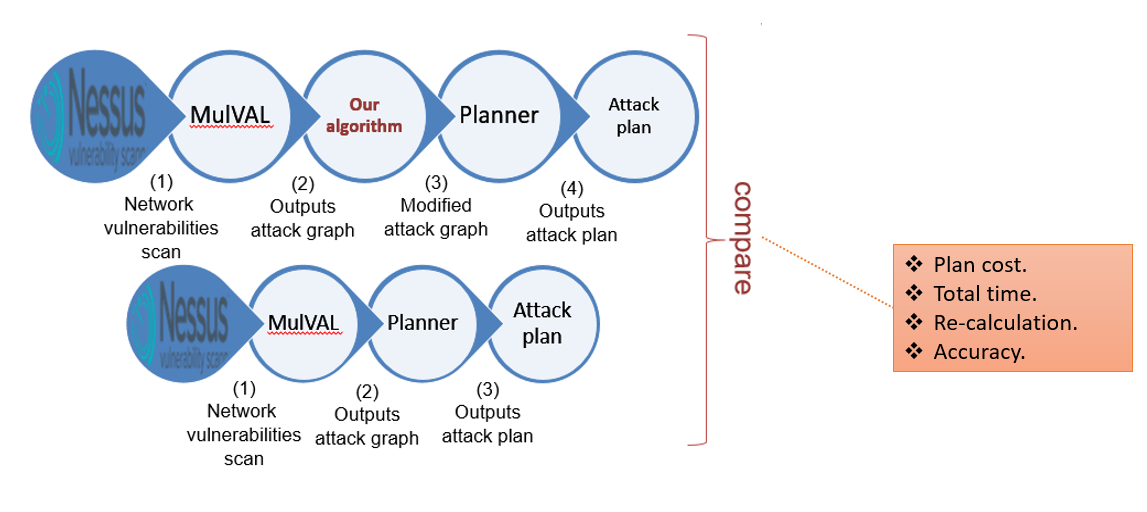}
\caption{Evaluation process}
\label{fig:Evaluation}
\end{figure}
%%%%%%%%%%%%%%%% end Figure %%%%%%%%%%%%%%%%%%%
\subsection{Proof of concept}\label{subsub: Proof of concepts}
We used Fred Cohen’s \emph{Deception ToolKit (DTK)} \cite{DTK} in order to generate deceptive information about a host (runs Linux OS). We used the guidelines detailed in \cite{DTKGuide}. Applying the DTK on a system causes it to appear to attackers as having a large number of widely known vulnerabilities. When the attacker issues a command or request the DTK results with a pre-defined response, in order to encourage the attacker to continue its exploration of the host.
The Deception tool kit is a set of C programs, Perl scripts, and additions/modifications to host configuration files.\\
\subsection{Defining evaluation measures } \label{subsection:parameters}
We evaluated our method using four evaluation measures:
\begin{enumerate}
\setlength{\itemsep}{0pt}
\setlength{\parskip}{0pt}
\item $p_1 =$ Number of recalculation - this number counts the number of times the attacker needs to recalculate its plan due to fake vulnerability. It is increased each time the adversary's attack plan includes of fake vulnerabilities and re-planning is needed.
\item $p_2 =$ Total time - total planning time (the time it took for the attacker to plan its path to reach the goal).
\item $p_3 =$ Relative increase of the attacker's cost - calculated as $\frac{total\; cost}{original\;cost}$.\\ The total cost, is the total cost the attacker invest during the attack of the network after the obfuscation, and the original cost is the total cost the attacker invest during the attack of the original network, with no obfuscation.\\
\item $p_4 =$ Precision - The number of fake assignments the attacker encounters in its path to the goal, relative to the number of fake assignments placed in the attack graph. More accurate the precision = $\frac{number\;of\;recalculation\;-\;1}{budget}$
\end{enumerate}
The procedure of calculating these measures is presented in Algorithm \ref{algo: Evaluating success measures}. 

\begin{algorithm}[H]
\caption{Evaluating success measures}
\label{algo: Evaluating success measures}
\begin{algorithmic}[1]
\renewcommand{\algorithmicrequire}{\textbf{Input:}}
\renewcommand{\algorithmicensure}{\textbf{Output:}}
\Require AG, numOfFakeAssignments $\gets$ number of fake assignments placed in the attack graph, originalCost
\Ensure reCalculation, relativeIncreaseOfCost, totalTime, precision
\State reCalculation, totalPlanCost, totalTime $\gets$ 0
\State $i \gets 1$
\State Let $Assignments = ((IP_1,v_1), (IP_2, v_2) ... (IP_m,v_n))  \Leftarrow  $ all the assignment of fake vulnerabilities to IPs.
\State timeStart = now()
\State AG' $\gets$ AA(AG,assignments)
\State $p_i=(u_1, u_2, ... u_m) \gets OPT(AG')$
\State timeEnd = now()
\State flag = True
\While{flag}
\State flag = False
\If{$ \exists k$, such that $u_k$ is a vertex created from some fake assignment $a \in Assignments$}
%\State Let $u_k$ be the first vertex which maintains the above.
\State flag = True
\State totalPlanCost = totalPlanCost + $\sum_{j=1}^{k} C(u_j)$
\State totalTime = totalTime + (timeEnd - timeStart)
\State currExpNodes $\gets$ expanded nodes during the current planning
\For {vertex $\in$ OPT(AG')}
\If {vertex != $u_k$}
\State cost($u_k$) = 0
\Else
\State break
\EndIf
\EndFor
\State timeStart = now()
\State AG' $\gets$ removeVertexFromAttackGraph(AG',a)
\State $p_i=(u_1, u_2, ... u_m) \gets OPT(AG')$
\State timeEnd = now()
\EndIf
\State $i \in i+1$
\State reCalculation = reCalculation +1
\EndWhile
\State precision $\gets \frac{reCalculation-1}{numOfFakeAssignments}$
\State relativeIncreaseOfCost $\gets \frac{total\; cost}{original\;cost}$
\State return reCalculation, relativeIncreaseOfCost, totalTime, precision
\end{algorithmic}
\end{algorithm}

\clearpage
\textbf{Correctness: } 
Due to the monotonic assumption \cite{ammann2002scalable} (mentioned in section \ref{sec:problem}), we assume that if an attacker gains some privileges in the network, it does not loose them. The evaluation measures are calculated considering this assumption. 
\\For calculating the evaluation measures we identified  two cases:
\begin{enumerate}
\item If the attacker tries to exploit a fake vulnerability, it invests the associated cost and then replans a path to the goal, taking into account the achievements made thus far. The parameters are calculated as follows, in each iteration:
\begin{itemize}
\item Total time is calculated as the sum of all the time it took to produce all the attack plans so far. 
\item Expected cost is calculated as the sum of:
\begin{itemize}
\item The costs of all of the vertexes in the attack plan that preceded the vertex which contains the fake vulnerability (including the vertex itself).
\item The total cost the attacker spent until this point (in all of its previous iterations).
\end{itemize}
\item Number of recalculations is computed as the number of attack plans that include fake vulnerabilities, so far, plus 1 (current).
\end{itemize}

\item If the attacker successfully reaches the goal, it means that the attack plan was devoid of fake vulnerabilities. In this case, the measures are calculated as follows:
\begin{itemize}
\item Total time is calculated as the sum of the time it took to plan the path to the target and the total time the attacker spent during all the attack plans so far. 
\item Expected cost is calculated as the sum of the costs of all of the vertexes in the attack plan and the total cost the attacker spent spent during all the attack plans so far. 
\item Number of recalculations is calculated as the number of recalculations so far plus 1.
\end{itemize}
\end{enumerate}
\par The relative increase of the attacker's cost calculated as the total cost (the total cost the attacker invest during the attack of the network after the obfuscation) encountered out of the original cost (the  total cost the attacker invest during the attack of the original network, with no obfuscation).  And mathematically, $\frac{total\; cost}{original\;cost}$.\\

\par The precision is calculated as the fraction of the number of fake vulnerabilities the attacker encountered out of the total number of fake vulnerabilities added to the attack graph ($\frac{reCalculation-1}{numOfFakeAssignments}$).\\

We compute the $planCost$, taking into account the fact that the adversary stops carrying out its attack plan when it is interrupted by a fake vulnerability. 
Therefore, the cost of the attack plan is the sum of all of the exploitable vulnerabilities, including the cost of trying to exploit the fake vulnerability.
\par In order to follow the monotonic assumption, in each iteration, we nullify the cost of all of the vertexes in the attack plan until the first fake vertex. From this vertex on, the attacker cannot execute the attack as planned or gain additional privileges.

\subsection{Results} \label{subsub: Results}
In this section we present the results of the two approaches: the random and the AI search approach. 
First, we evaluate each of the two approaches separately, and then we compare the two.\\\\
\textbf{Platform.} In all the experiments presented in this section we used a virtual machine who runs an Ubuntu 12.04.5 with a Linux kernel 3.13.0-95. System with a 4 Virtual CPU, 2Tb Disk and 8G RAM.\\
The resources were fully dedicated to our experiments.\\
On each experiment, we followed the steps described in section \ref{sec:method} to generate an attack graph and assign fake vulnerabilities.

\subsubsection{Random approach}
\paragraph{Implementation:}\mbox{}\\
The purpose of this experiments were as follows:
\begin{enumerate}
\item Check the feasibility of the attack graph obfuscation.
\item Check the precision and scalability of the random approach.
\end{enumerate}
\begin{itemize}
\item \textbf{Data set:}\mbox{}\\
This experiments conducted to a sub-network ($NET_{80}$) of the network described in section \ref{sec:data} that consists of 80 hosts (see the attack graph in Figure \ref{fig:dt80}).\\
%%%%%%%%%%%%%%%% begin Figure %%%%%%%%%%%%%%%%%%%
\begin{figure}
  \centering
  \includegraphics[width=.6\linewidth]{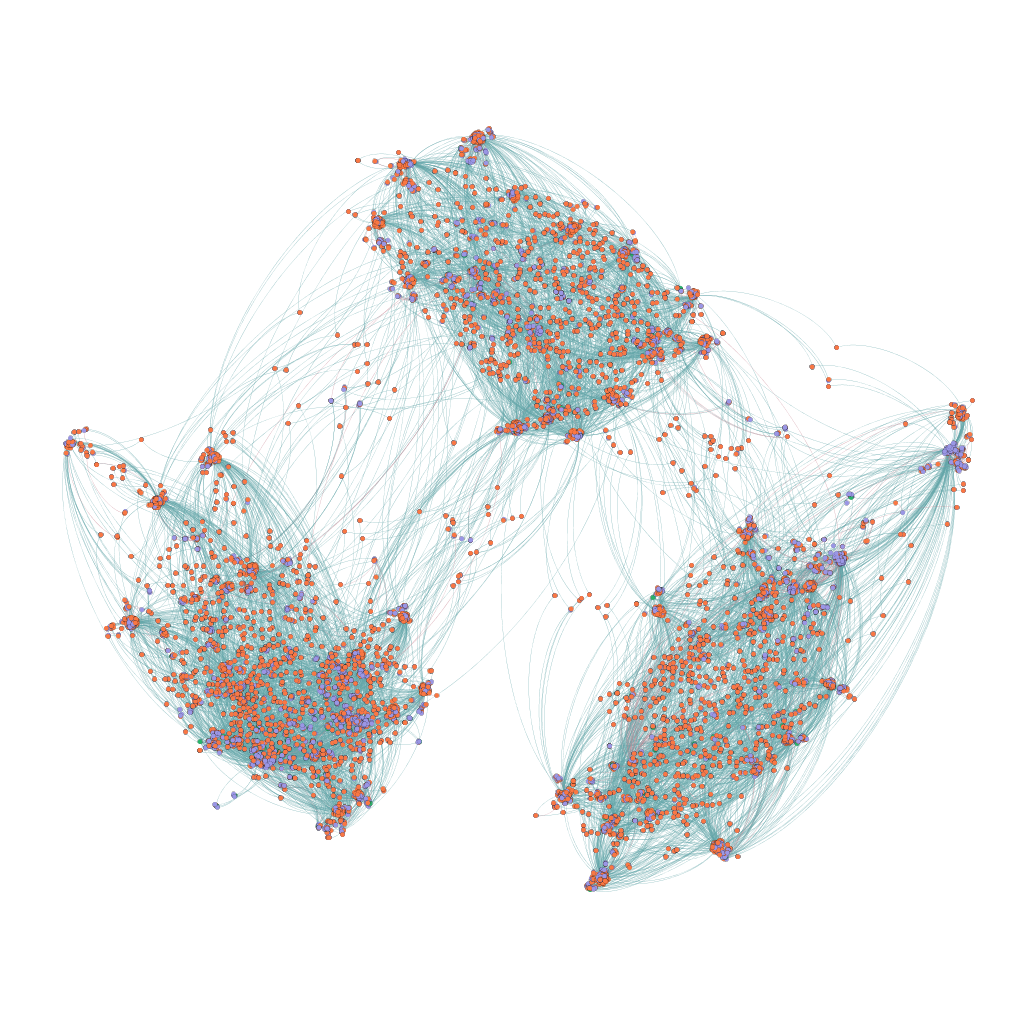}
  \caption{Attack graph of $NET_{80}$}
  \label{fig:dt80}
\end{figure}
%%%%%%%%%%%%%%%% end Figure %%%%%%%%%%%%%%%%%%%

\item \textbf{Parameters:}
\begin{itemize}
\item \textbf{Budget:}
\begin{enumerate}
\item 10\% of the network (8 hosts) chosen as deceptive.
\item 30\% of the network (24 hosts) chosen as deceptive.
\item 50\% of the network (40 hosts) chosen as deceptive.
\end{enumerate}
\end{itemize}
\end{itemize}
For each budget amount, we repeated the experiments five times.  
At each trial, the number of chosen fake vulnerabilities for each host is randomly chosen.  
At each trial, a new obfuscated attack graph is constructed.
The results in the following section present the average for all the trials conducted for each evaluated budget amount including zero (no deception).
Table \ref{table:1} provides the attack graphs' data, which includes the attack graphs' average number of vertexes, edges, and vulnerabilities relative to the budget.
\begin{table}[ht]
\centering
\begin{tabular}{||c | c | c | c||} 
 \hline
 Deceptive IPs & Vertexes & Edges & Vulnerabilities\\ [0.5ex] 
 \hline\hline
 0\% & 10147 & 16591 & 140\\
 \hline
 10\% & 10815 & 18421 & 164.2\\
 \hline
 30\% & 12392.8 & 22591 & 220.6\\
 \hline
 50\% & 13936.8 & 26734.4 & 257.2\\ [1ex] 
 \hline
\end{tabular}
\caption{Attack graph data}
\label{table:1}
\end{table}

\paragraph{Results:}\mbox{}\\
We present results for the attacker cost during the planning and attack phases.
The results are presented in Figure \ref{fig:results} and Figure \ref{fig:randomResultsPrecision} followed with further explanations.\\
%%%%%%%%%%%%%%%% begin Figure %%%%%%%%%%%%%%%%%%%
\begin{figure}[ht]
\centering
\centering
\includegraphics[width=.9\linewidth]{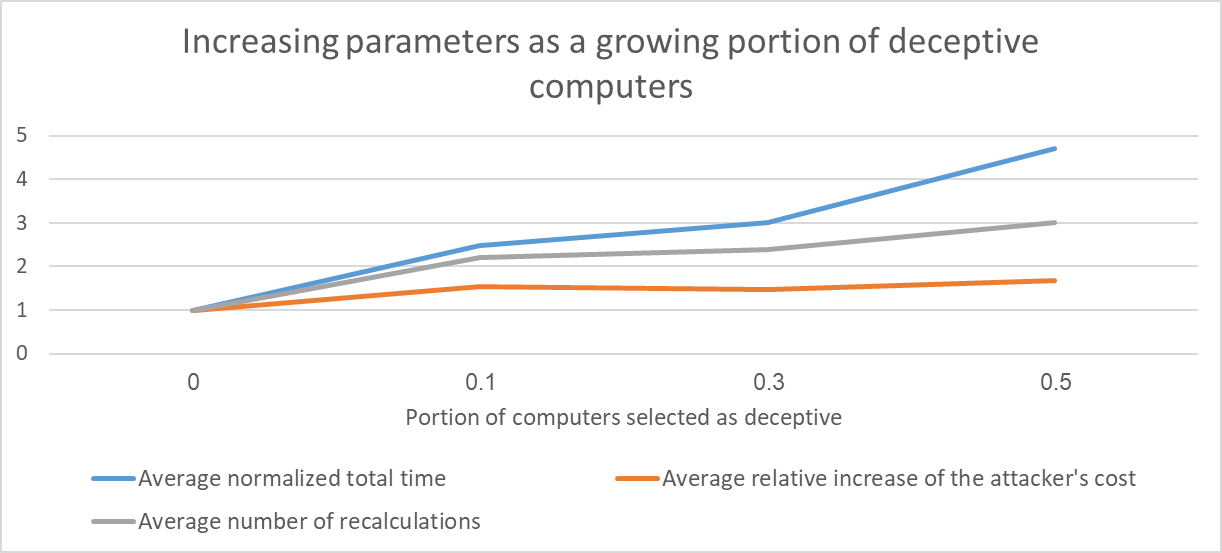}
\caption{Random approach results}
\label{fig:results}
\end{figure}
%%%%%%%%%%%%%%%% end Figure %%%%%%%%%%%%%%%%%%%

%%%%%%%%%%%%%%%% begin Figure %%%%%%%%%%%%%%%%%%%
\begin{figure}[ht]
\centering
\centering
\includegraphics[width=.9\linewidth]{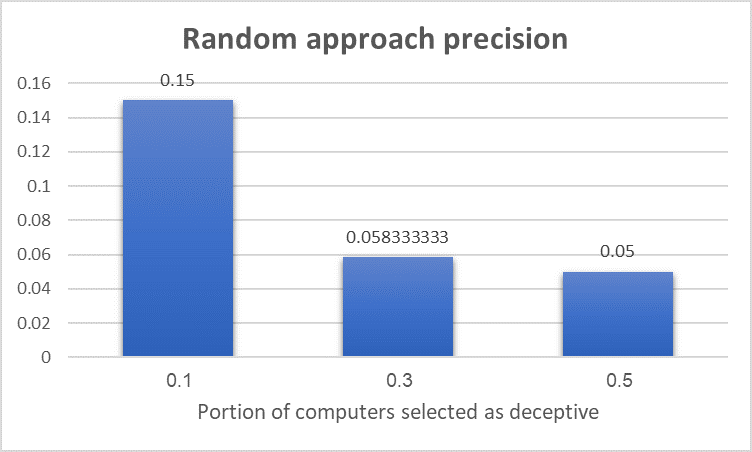}
\caption{Random approach precision}
\label{fig:randomResultsPrecision}
\end{figure}
%%%%%%%%%%%%%%%% end Figure %%%%%%%%%%%%%%%%%%%

In Figure \ref{fig:randomResultsPrecision}, we can see the precision of the random approach. We can see that as the budget increase, the total cost increase, but not in the same proportion as the number of computers defined as deceptive. Thus, we can see that even by choosing 10\% of the computers as deceptive, we can get a result close to the result obtained from the choice of 50\% of the computers as deceptive.

In Figure \ref{fig:results}, we can see the values of the attacker's total time invested, relative increase of the attacker's cost and the number of recalculation, relative to the portion of computers chosen as deceptive. We can see that as the portion of deceptive computers increased, the attacker's effort increased. 
Next we discuss the results depicted in Figure \ref{fig:results}.
\begin{itemize}
\item We measure the attacker's effort during the planning phase (\(p_2\)). 
Table \ref{table:2} shows the increase in planning time relative to the baseline.
\begin{table}[ht]
\centering
\begin{tabular}{||c | c ||} 
 \hline
 \% of deceptive computers & total time\\ [0.5ex] 
 \hline\hline
  10\% & 2.5 \\ 
 \hline
 30\% & 3 \\
 \hline
 50\% & 4.7\\
 \hline
\end{tabular}
\caption{Planing time relative to baseline}
\label{table:2}
\end{table}

It can be observed that even for only 10\% deceptive computers the planning time is more than double, and it grows with the number of deceptive computers.

\item We evaluate the attacker's effort during the attack phase through the relative increase of the attacker's cost (\(p_3\)).\\ 
Table \ref{table:4} presents the relative increase of the attacker's cost.\\
\begin{table}[ht]
\centering
\begin{tabular}{||c | c ||} 
 \hline
 \% of deceptive computers & relative increase of the attacker's cost\\ [0.5ex] 
 \hline\hline
 10\% & 1.46\\ 
 \hline
 30\% & 1.53 \\
 \hline
 50\% & 1.6\\
 \hline
\end{tabular}
\caption{Total attack path cost relative to baseline}
\label{table:4}
\end{table}

As observed from Table \ref{table:1} and Table \ref{table:4} it can be seen that just a small amount of fake vulnerabilities is needed to effectively waste the adversary's resources, and more specifically that by randomly adding vulnerabilities to less than 30\% of the computers in the network can cost the adversary almost 1.5 times more.

\item We evaluate the number of recalculations, which is the number of times the attacker needs to recalculate its plan due to fake vulnerabilities. Table \ref{table:Number of recalculations} presents the number of recalculation relative to the number of deceptive computers chosen.\\
\begin{table}[ht]
\centering
\begin{tabular}{||c | c ||} 
 \hline
 \% of deceptive computers & Number of recalculations\\ [0.5ex] 
 \hline\hline
 10\% & 2.2\\ 
 \hline
 30\% & 2.4 \\
 \hline
 50\% & 3\\
 \hline
\end{tabular}
\caption{Number of recalculations}
\label{table:Number of recalculations}
\end{table}

As observed from Table \ref{table:Number of recalculations} it can be seen that just a small amount of fake vulnerabilities is needed to effectively duplicate the needed recalculations. More specifically that by randomly adding vulnerabilities to less than 10\% of the computers in the network the number of recalculations duplicates.
\end{itemize}

\clearpage

\subsubsection{AI search approach} \label{section:AI search approach}
\paragraph{}{Implementation:}\mbox{}\\
This experiment conducted in order to:
\begin{enumerate}
\item Examine which algorithm parameters provide better results in terms of run-time and precision.
\item Determine whether the AI search approach has better precision than the random approach.
\item Determine whether the AI search approach is feasible for large networks.
\end{enumerate}
\begin{itemize}
\item \textbf{Data set:}\mbox{}\\
For evaluation needs, and in order to test our approach on challenging and interesting networks, we created four sub-networks from the network described in section \ref{sec:data}, using the same topology:
\begin{enumerate}
\item $NET_4$ - a subset network which created with 4 hosts. \\ You can see Figure \ref{fig:DT_4} the constructed attack graph (AG). \label{network4}\\
\textbf{Number of vertexes: } 96\\
\textbf{Number of edges: } 136\\
\item $NET_{10}$ - a subset network which created with 10 hosts \\ You can see Figure \ref{fig:DT_10} the constructed attack graph (AG).\label{network10}\\
\textbf{Number of vertexes: } 486\\
\textbf{Number of edges: } 728\\
\item $NET_{20}$ - a subset network which created with 20 hosts. \\ You can see Figure \ref{fig:DT_20} the constructed attack graph (AG).\label{network20}\\
\textbf{Number of vertexes: } 2179\\
\textbf{Number of edges: } 3685\\
\end{enumerate}

\begin{figure}[ht] 
    \centering
    \includegraphics[width=.5\linewidth]{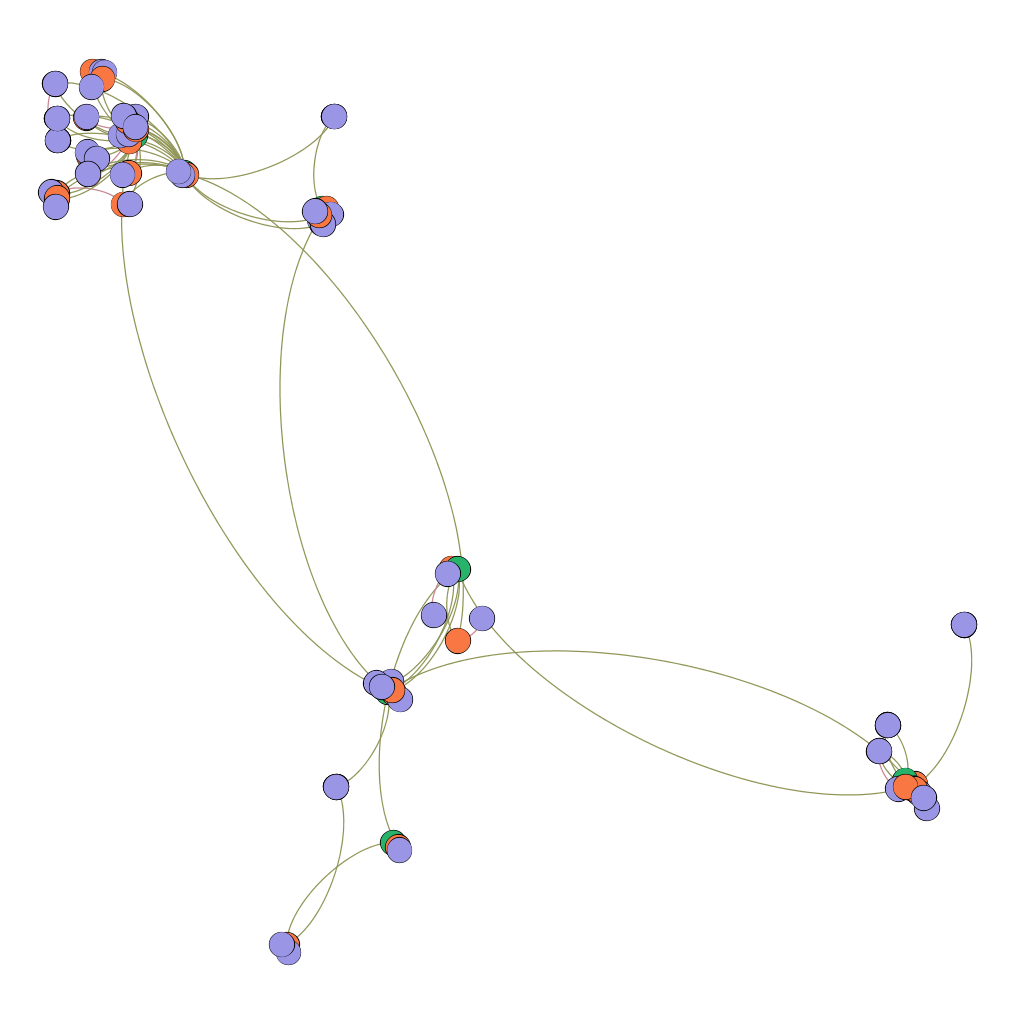} 
    \caption{Attack graph of $NET_{4}$}
    \label{fig:DT_4}
\end{figure}%%
\begin{figure}[ht] 
    \centering
    \includegraphics[width=.5\linewidth]{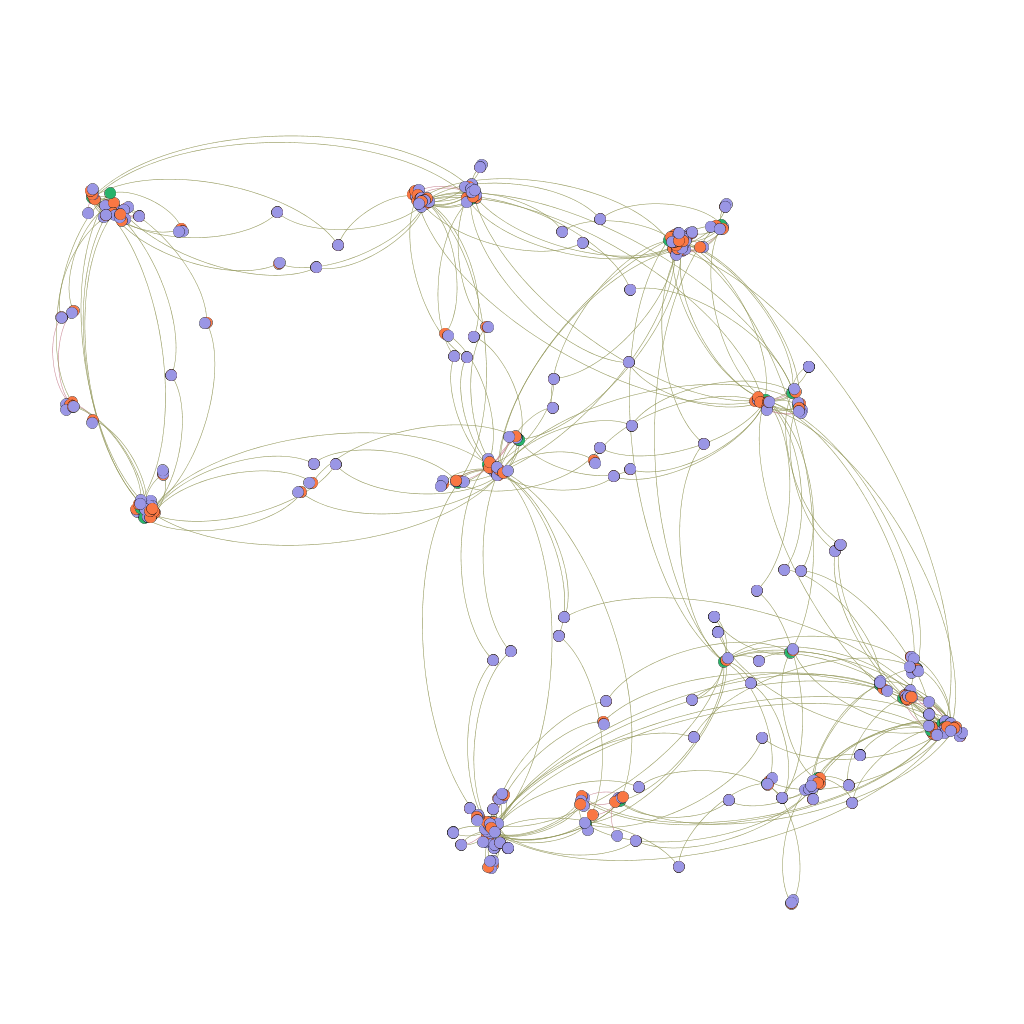} 
    \caption{Attack graph of $NET_{10}$} 
    \label{fig:DT_10}
\end{figure} 
\begin{figure}[ht] 
    \centering
    \includegraphics[width=.5\linewidth]{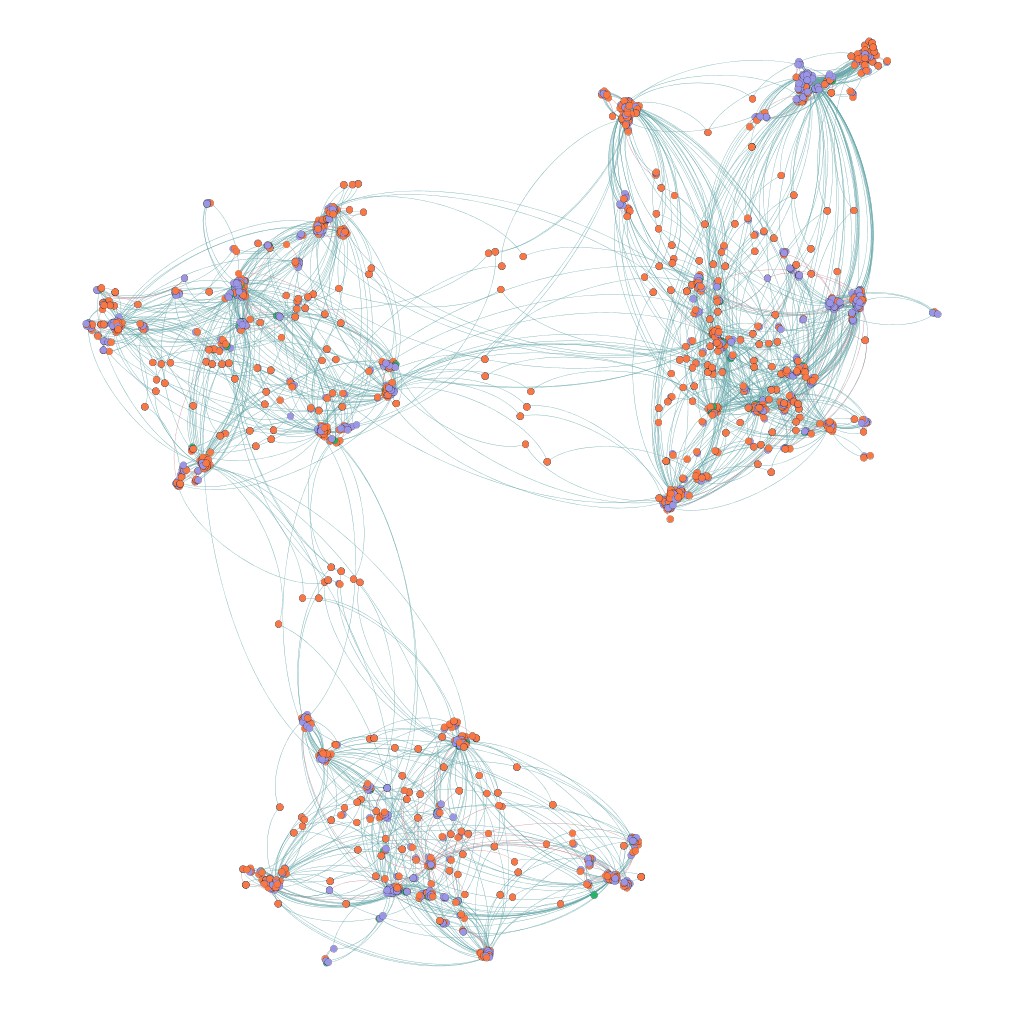} 
    \caption{Attack graph of $NET_{20}$} 
    \label{fig:DT_20}
\end{figure}

\item \textbf{Parameters:}\mbox{}\\
For the AI search approach we compared numerous configurations of the parameters (described in section \ref{subsub: AI search based approach}):
\begin{itemize}
\item Two search algorithms:
\begin{enumerate}
\item A*
\item DF-BnB
\end{enumerate}
\item Three assignments ordering of the candidates list:
\begin{enumerate}
\item High utility first ordering
\item Shortest path first ordering
\item Random ordering
\end{enumerate}
\item Two heuristic evaluation functions:
\begin{enumerate}
\item $h_1$ (utility upper bound)
\item $h_2$ (admissible utility upper bound)
\end{enumerate}
\item Budgets:
\begin{enumerate}
\item Budget of 3 assignments
\item Budget of 5 assignments
\end{enumerate}
\end{itemize}

\end{itemize}

\paragraph{Results}
\begin{enumerate}
  \item \textbf{Heuristic evaluation functions} - the heuristics evaluation functions differ in the number of expanded nodes during the search, which effect the total time and the achieved total cost (which reflect the amount of effort the adversary is putting into it's attack).
\begin{itemize}
\item \textbf{Precision -} which is the number of fake vulnerabilities the attacker try to exploit during the attack, relative to the budget. In other words, the precision is equal to $\frac{number\;of\;recalculation\;-\;1}{budget}$, due to the fact that the number of recalculations is the number of attack plan consists fake vulnerabilities + 1 (the attack plan which consists valid path to the goal).
\\\textbf{The precision is 1, for each of the heuristics.} Which means that all the assignments chosen by the search algorithm were deceptive to the attacker. In practice, each fake assignment, was tested by the attacker. This maximizes the adversary's setback while trying to reach its goal in the targeted network. But, if the budget is greater than the assignments needed, our algorithm choose the minimum number of assignments which can optimally increase the attacker's efforts (see Figure \ref{fig:budgetUsage}). For example, in $NET_{20}$, when the budget was 5, the best fake vulnerability assignment were the same as the optimal assignment for a budget of 4 ($=5*0.8$), so there was no need for 5 assignments.
%%%%%%%%%%%%%%%% begin Figure %%%%%%%%%%%%%%%%%%%
\begin{figure}[ht]
\centering
\centering
\includegraphics[width=.7\linewidth]{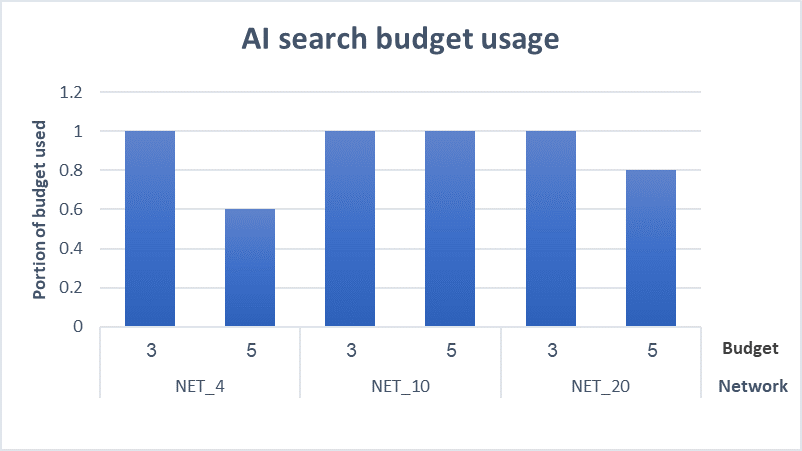}
\caption{Budget Usage}
\label{fig:budgetUsage}
\end{figure}
%%%%%%%%%%%%%%%% end Figure %%%%%%%%%%%%%%%%%%%

\item \textbf{Relative increase of the attacker's cost} \label{explanation: Relative increase} = $\frac{total\; cost}{original\;cost}$.\\ The total cost, is the total cost the attacker invest during the attack of the network after the obfuscation, and the original cost is the total cost the attacker invest during the attack of the original network, with no obfuscation.\\
We can see that the non-admissible heuristic can give a non-optimal solution, which is shown in the results presented in Figure \ref{fig:heuristicsTotalCost}. In $NET_{20}$ network, the AI search, using the non-admissible heuristic precision was 1 but the chosen hosts did not gave an optimal solution. 
%%%%%%%%%%%%%%%% begin Figure %%%%%%%%%%%%%%%%%%%
\begin{figure}[ht]
\centering
\centering
\includegraphics[width=.9\linewidth]{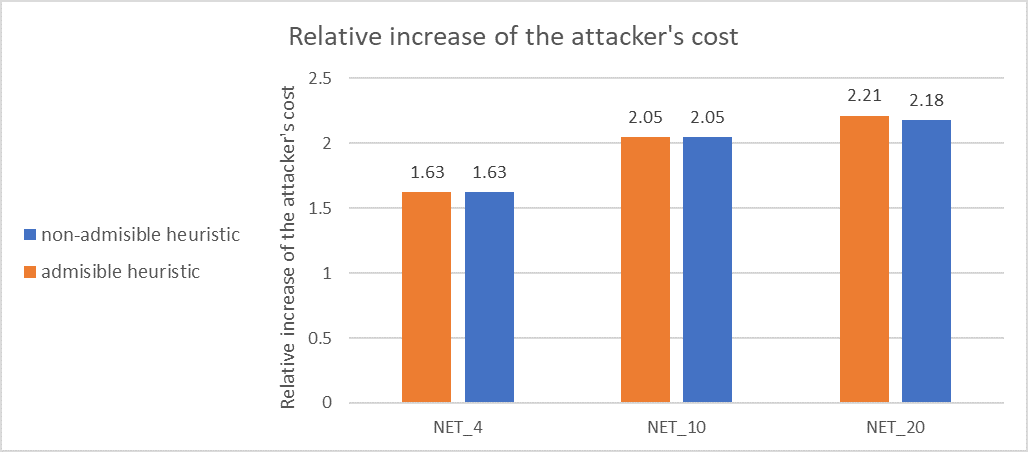}
\caption{Relative increase of the attacker's cost for each heuristic}
\label{fig:heuristicsTotalCost}
\end{figure}
%%%%%%%%%%%%%%%% end Figure %%%%%%%%%%%%%%%%%%%

\item \textbf{Expanded nodes -} The difference between the expanded nodes, for each of the heuristics is negligible. For the networks $NET_{4}$ and $NET_{10}$ there were no difference in the number of expanded nodes, the only difference were for the network $NET_{20}$.\\
We omitted the assignments random ordering, due to the fact that the random selection can not provide valid information in this subject. Figure \ref{fig:heuristicExpandedNodes20}, present the number of expanded nodes for network $NET_{20}$.
\\The expanded nodes where higher in the case of using the admissible heuristic with DFBnB search algorithm, for budget of five hosts, by using the high utility first ordering.\\
The actual difference is in 7\% more, for the admissible heuristic, which is negligible.
%%%%%%%%%%%%%%%% begin Figure %%%%%%%%%%%%%%%%%%%
\begin{figure}[ht]
\centering
\centering
\includegraphics[width=1\linewidth]{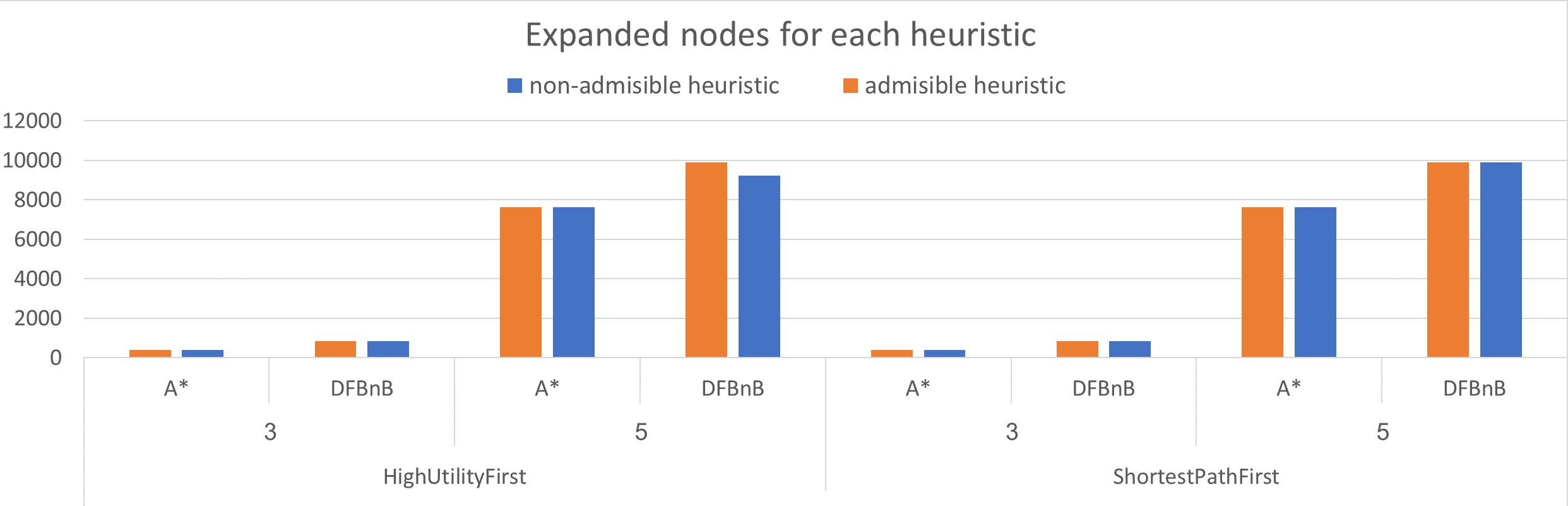}
\caption{Heuristics expanded nodes, in network $NET_{20}$}
\label{fig:heuristicExpandedNodes20}
\end{figure}
%%%%%%%%%%%%%%%% end Figure %%%%%%%%%%%%%%%%%%%

\item \textbf{Search total time - } the search total time, is a direct result of the number of the search expanded nodes. Same here, the differences in the total time between the admissible heuristic to the non-admissible is negligible. For the networks $NET_{4}$ and $NET_{10}$ there were no difference in the total time, the only difference were for network $NET_{10}$.\\
Figure \ref{fig:heuristicsTotalTime_20}, present the total search time for the network $NET_{20}$.
%%%%%%%%%%%%%%%% begin Figure %%%%%%%%%%%%%%%%%%%
\begin{figure}[ht]
\centering
\centering
\includegraphics[width=1\linewidth]{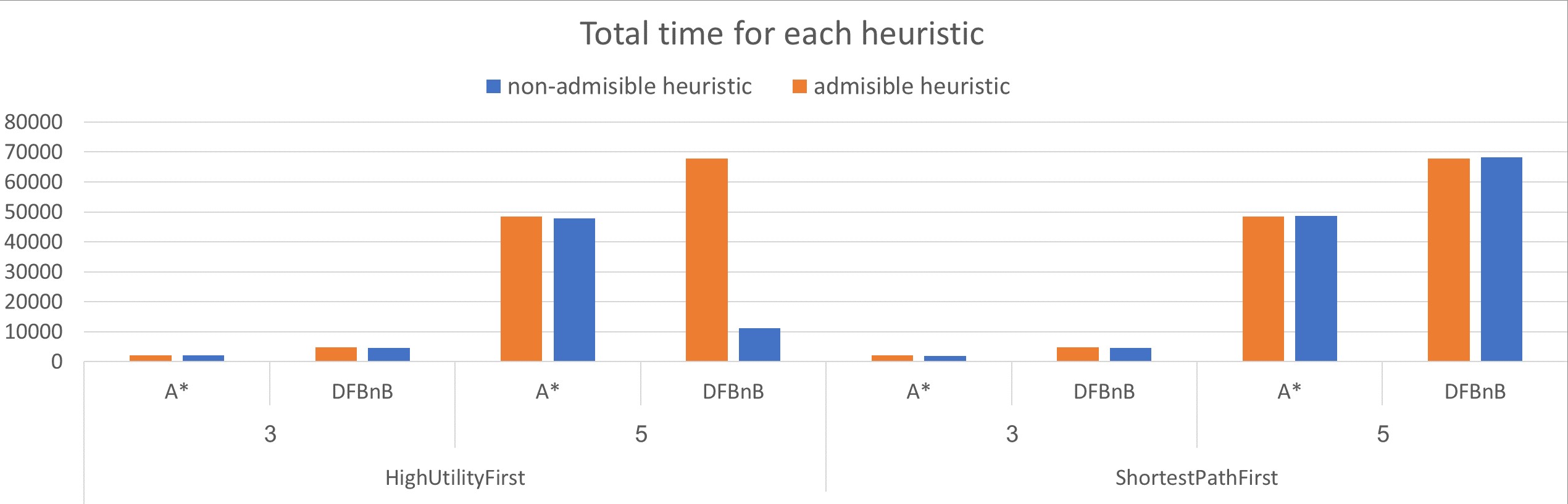}
\caption{Heuristics total time, in network $NET_{20}$}
\label{fig:heuristicsTotalTime_20}
\end{figure}
%%%%%%%%%%%%%%%% end Figure %%%%%%%%%%%%%%%%%%%

\end{itemize}
\item \textbf{Assignments ordering} - the purpose of the various heuristics for candidates ordering (in the ordered assignment candidates list AC(V)) is to reduce the \textbf{number of nodes created during the search.} As a result, the total time of the search will be reduced.\\
From the results (see Figure \ref{fig:OrderingExpandedNodes}) we can see that usually the random approach performs worse than the proposed ordering heuristics - shortest path first and high utility first. whose results are usually similar.
%%%%%%%%%%%%%%%% begin Figure %%%%%%%%%%%%%%%%%%%
\begin{figure}[ht]
\centering
\centering
	\includegraphics[width=1\linewidth]{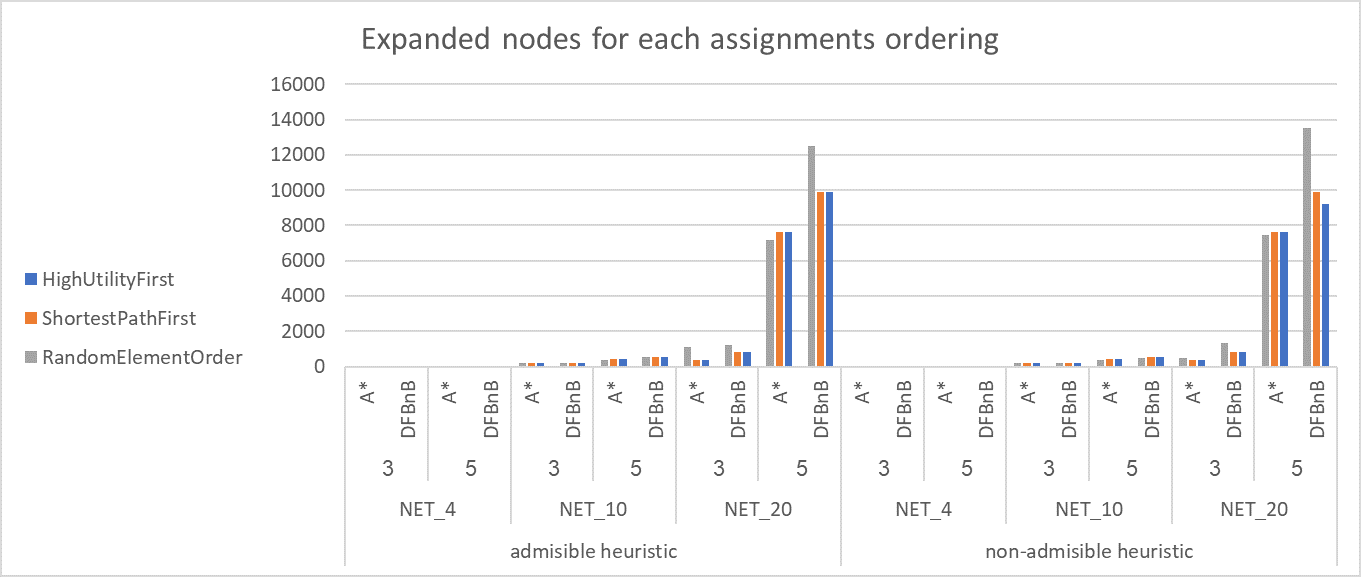}
\caption{Expanded nodes by candidate assignments ordering type}
\label{fig:OrderingExpandedNodes}
\end{figure}
%%%%%%%%%%%%%%%% end Figure %%%%%%%%%%%%%%%%%%%

\item \textbf{Search algorithms comparison} - 
\begin{itemize}
\item \textbf{Expanded nodes -} we can see that A* is a better search algorithm for solving this problem, because the number of expanded nodes is larger for the DFBnB algorithm, than the A* algorithm (as shown in Figure \ref{fig:searchAlgorithmExpanded}). We can see that for all network sizes (4, 10 and 20 hosts) and for all the tested budgets (3 or 5), the number of expanded nodes using the A* is smaller than the number of expanded nodes using DFBnB search algorithm. 
%%%%%%%%%%%%%%%% begin Figure %%%%%%%%%%%%%%%%%%%
\begin{figure}[ht]
\centering
\centering
\includegraphics[width=1\linewidth]{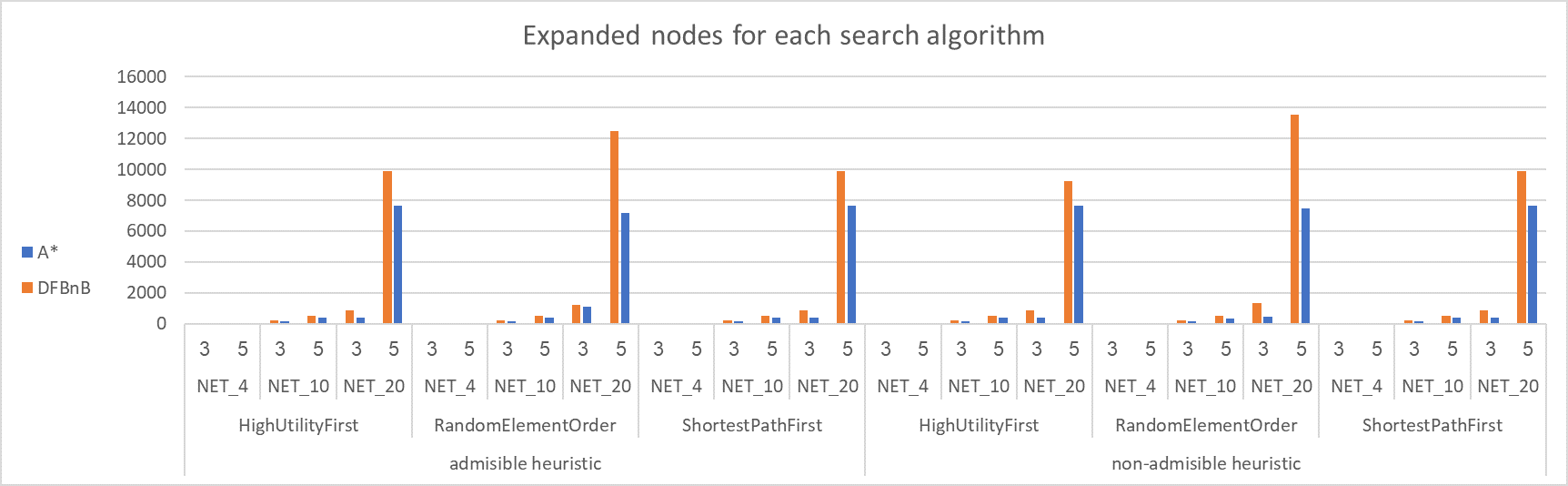}
\caption{Expanded nodes for each search algorithm}
\label{fig:searchAlgorithmExpanded}
\end{figure}
%%%%%%%%%%%%%%%% end Figure %%%%%%%%%%%%%%%%%%%

\item \textbf{Total search time -} we can see that as a result of the number of expanded nodes, the total time it took to find the optimal solution in A* is better than in DFBnB (as shown in Figure \ref{fig:searchAlgorithmTime}).
%%%%%%%%%%%%%%%% begin Figure %%%%%%%%%%%%%%%%%%%
\begin{figure}[ht]
\centering
\centering
\includegraphics[width=1\linewidth]{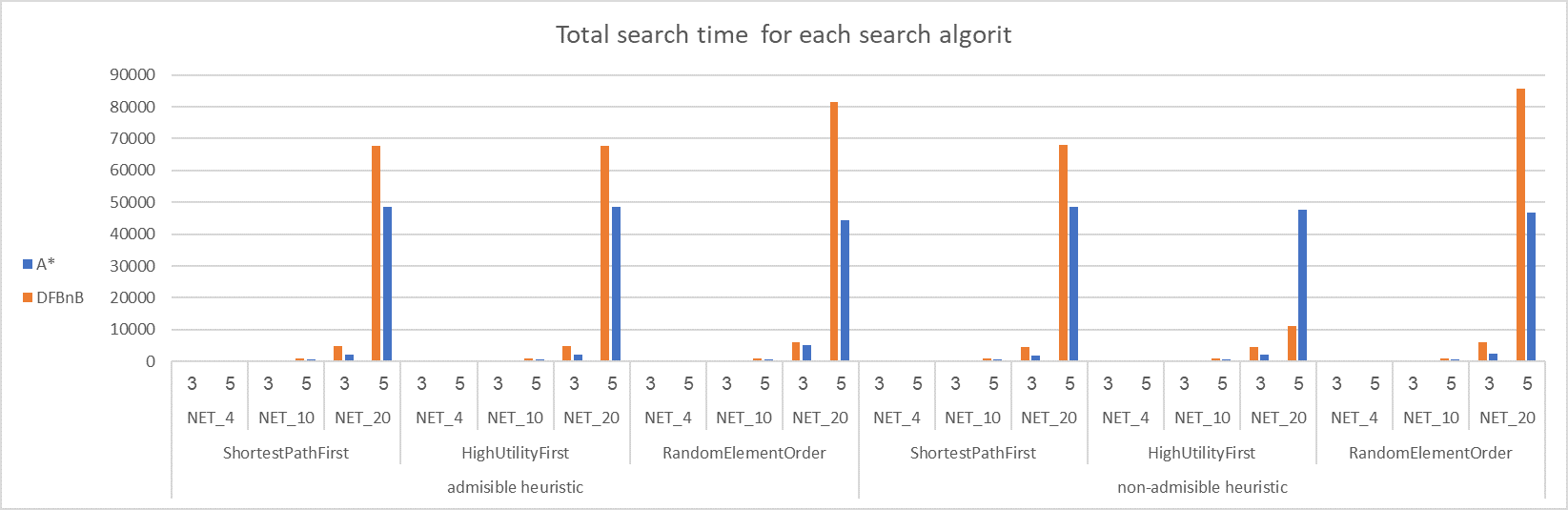}
\caption{Total search time for each search algorithm}
\label{fig:searchAlgorithmTime}
\end{figure}
%%%%%%%%%%%%%%%% end Figure %%%%%%%%%%%%%%%%%%%

\end{itemize}
From the results it can be concluded that in order to solve this problem, A* preform better than DFBnB.

\end{enumerate}

\subsubsection{Random approach vs. AI search approach}
\paragraph{Implementation:}\mbox{}\\
In this section we demonstrate that the AI approach is preferable to the random approach in 
utilizing the budget in the most efficient way.

\begin{itemize}
\item \textbf{Data set:}
For evaluation needs, and in order to test our approach on challenging and interesting networks, we created four sub-networks from the network described in section \ref{sec:method}, using the same topology:
\begin{enumerate}
\item $NET_{10}$ - as shown in Figure \ref{fig:DT_10} (page \pageref{network10}).
\item $NET_{20}$ - as shown in Figure \ref{fig:DT_20} (page \pageref{network20}).
\item $NET_{50}$ - a subset network which created with 50 hosts.\\
You can see in Figure \ref{fig:dt50} the constructed attack graph (AG).\\
Number of vertexes: 10,203\\
Number of edges: 16,621\\

%%%%%%%%%%%%%%%% begin Figure %%%%%%%%%%%%%%%%%%%
\begin{figure}
  \centering
  \includegraphics[width=.4\linewidth]{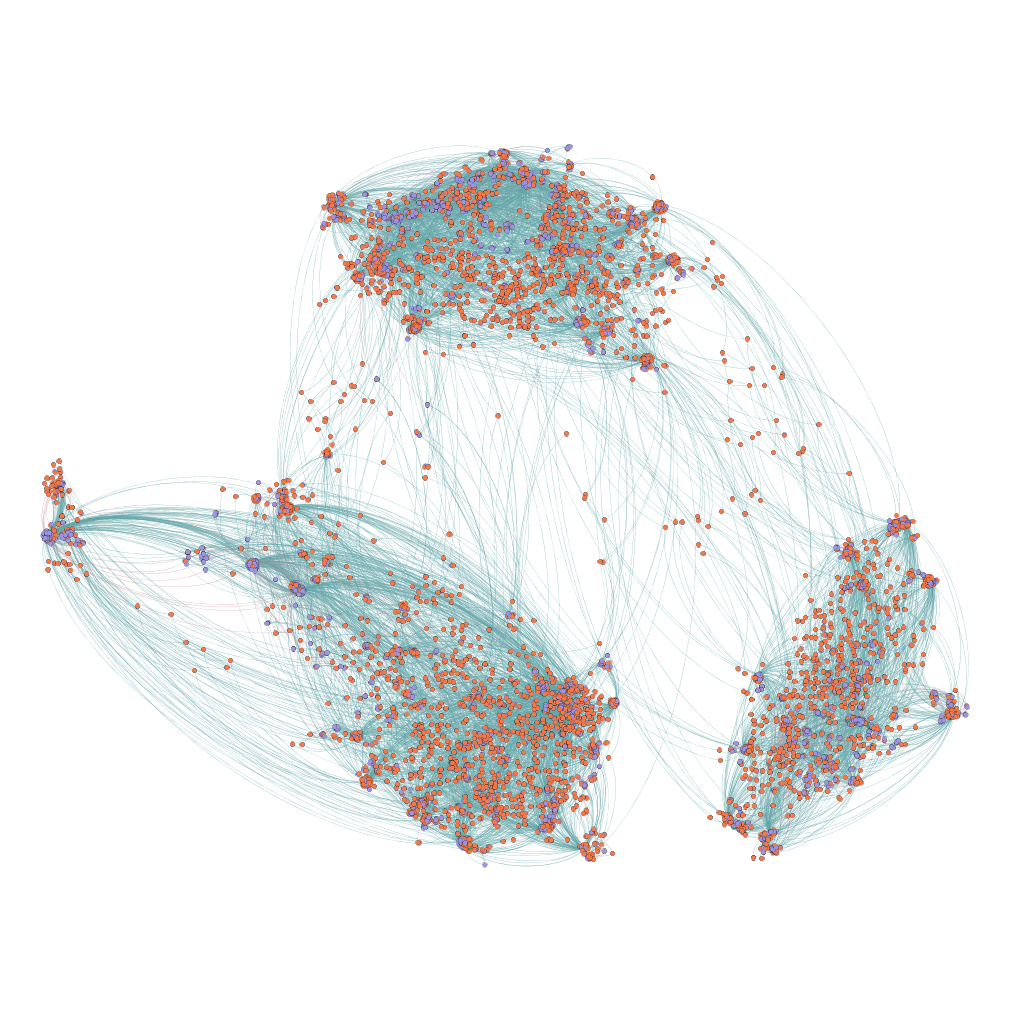}
  \caption{Attack graph of the sub network with 50 hosts}
  \label{fig:dt50}
\end{figure}
%%%%%%%%%%%%%%%% end Figure %%%%%%%%%%%%%%%%%%%

\item $NET_{80}$ - as shown in Figure \ref{fig:dt80} (page \pageref{fig:dt80}). The attack graph consists of 10,147 vertexes and 16,591 edges.
\end{enumerate}
\item \textbf{Parameters:}
\begin{itemize}
\item \textbf{AI approach parameters:}
\begin{enumerate}
\item \textbf{Budget} - We started the trials on the AI approach with small budgets in order to demonstrate that the impact on the attacker (relative increase of the attacker's cost) can be significant, even with a small budget. Then, we examine what should be the provided budget as input to the random algorithm, in order to reach the same results.
This enables us to show the significant differences between the two different algorithms, even for low budgets. 
\item  \textbf{Search algorithm} - We chose the A* search algorithm, due to its effectiveness in terms of runtime, as seen in the results presented in section \ref{section:AI search approach}. 
\item  \textbf{Assignments ordering of the candidates list} - We chose the high utility first ordering, due to its effectiveness in terms of runtime, as seen in the results presented in section \ref{section:AI search approach}. 
\item  \textbf{Heuristic evaluation function} - We chose the admissible utility upper bound evaluation function, due to the fact it provides an optimal solution and is not poor in terms of running time, as shown in the results presented in section \ref{section:AI search approach}. 
\end{enumerate}

\item \textbf{Random approach parameters:}
\begin{enumerate}
\item \textbf{Budget} - We examine all the possible budgets from 1 deceptive computer to $n-1$, where n is number of hosts in the network. Then we matched the results to the results given by the AI approach. 
\item \textbf{Trials} - For each budget, we repeated the experiments five times. The results in the following section present the average for all the trials conducted for each budget. Therefore, at each trial, a new obfuscated attack graph is constructed. \\ 
At each trial, the number of chosen fake vulnerabilities for each host is randomly chosen.
\end{enumerate}
\end{itemize}
\end{itemize}

\paragraph{Results}\mbox{}\\
In this section we demonstrate that the AI approach is preferable to random approach in terms of optimal results.\\
The two approaches has clear trade offs:
\begin{itemize}
\item \textbf{Total search time -} As in all cases and since the random approach have no "smart" choices, the random approach is much faster than the AI search algorithm.
\item \textbf{Relative increase of the attacker's cost -} The effect on the attacker described and explained in page \pageref{explanation: Relative increase} as the relative increase of the attacker's cost.\\
Unlike the random approach, by using the AI search approach we can maximize the given budget for hardening the given network.\\
We can see from Table \ref{table: AI VS. Random}, that in order to get the same effect on the attacker, the AI search can provide a solution with much less deceptive computers and much less fake vulnerabilities. \\
We can see that the number of computers chosen as deceptive in the random approach is constantly very high in comparison to the number of computers needed to be deceptive in the AI search approach. \\
Furthermore, as expected when the number of hosts in the network increase, "guessing" the right deceptive computers is much harder.\\
The number of computers chosen as deceptive is at least 1.6 times higher in the random approach, and as the network grows it gets up to 7.3 times more.\\
The number of fake vulnerabilities assignments is at least 14 times higher in the random approach, and as the network grows it gets up to 73 times more.

\end{itemize}

\begin{table}[ht] 
\caption{AI search approach VS. Random approach}
\label{table: AI VS. Random}
  \begin{tabular}{|l|l|p{0.9in}|p{0.7in}|p{0.9in}|p{0.8in}|}
  \hline
    Approach & Network & Relative increase of the attacker's cost & Number of deceptive IPs & Number of fake vulnerabilities added  & Attackers total planning time (S)\\ \hline \hline \hline
    AI search & $NET_{10}$ & 1.8 &  \textcolor{ForestGreen}{3} & \textcolor{ForestGreen}{3} & 0.22\\ \hline
    Random & $NET_{10}$ & 1.8 & \textcolor{Red}{5} & \textcolor{Red}{42} & 0.196	\\ \hline \hline
    AI search & $NET_{10}$ & 2.3 & \textcolor{ForestGreen}{5} & \textcolor{ForestGreen}{5} & 0.3 \\ \hline
    Random & $NET_{10}$ & 2.3 & \textcolor{Red}{9} & \textcolor{Red}{103} & 0.26 \\ \hline \hline \hline
    
    AI search & $NET_{20}$ & 1.7 & \textcolor{ForestGreen}{2} & \textcolor{ForestGreen}{2} & 2.58 \\ \hline
    Random & $NET_{20}$ & 1.7 & \textcolor{Red}{7} & \textcolor{Red}{60} & 2.527 \\ \hline \hline
    AI search & $NET_{20}$ & 2.1 & \textcolor{ForestGreen}{3} & \textcolor{ForestGreen}{3} & 4.22 \\ \hline
    Random & $NET_{20}$ & 2.1 & \textcolor{Red}{10} & \textcolor{Red}{94} & 2.86 \\ \hline \hline \hline
    
    AI search & $NET_{50}$ & 1.6 & \textcolor{ForestGreen}{2} & \textcolor{ForestGreen}{2} & 55.8 \\ \hline
    Random & $NET_{50}$ & 1.6 & \textcolor{Red}{14} & \textcolor{Red}{148} & 43.544 \\ \hline \hline
    AI search & $NET_{50}$ & 2.1 & \textcolor{ForestGreen}{4} & \textcolor{ForestGreen}{4} &  164.5 \\ \hline
    Random & $NET_{50}$ & 2.1 & \textcolor{Red}{22} & \textcolor{Red}{238} &  156.33 \\ \hline \hline \hline
    
    AI search & $NET_{80}$ & 1.7 & \textcolor{ForestGreen}{2} & \textcolor{ForestGreen}{2} &  37.58\\ \hline
    Random & $NET_{80}$ & 1.7 & \textcolor{Red}{12} & \textcolor{Red}{122} & 34.492 \\ \hline \hline
    AI search & $NET_{80}$ & 2.1 & \textcolor{ForestGreen}{3} & \textcolor{ForestGreen}{3} & 110.7\\ \hline
    Random & $NET_{80}$ & 2.1 & \textcolor{Red}{22} & \textcolor{Red}{219} &  100.58 \\ \hline \hline
    
  \end{tabular}
\end{table}

%%%%%%%%%%%%%%%%%%%%%%%%%%%%%%%%%%%%%%%%%%%%%%%%%%%%%%%%%%%%%%%%%%%
%% related work
%%%%%%%%%%%%%%%%%%%%%%%%%%%%%%%%%%%%%%%%%%%%%%%%%%%%%%%%%%%%%%%%%%%

\section{Related work}
\label{sec:related}

The approach presented in this research utilizes a defense strategy that provides false information to attackers. 
The falsified information is optimized based on the attack graph representation of the protected network.  
Next we discuss the major topics related to this strategy including: deception, attack graph games, obfuscation, and others.

\subsection{Attack Graph Analysis} 
\label{sub: Attack plans}

Simulating an attacker can be done by multiple ways: using a heuristic search, planner, greedy algorithm and more. Planning an attack can be online or offline. When the environment is known, planning can be done offline, meaning that solutions can be found and evaluated prior to execution. 
But, in the real world the attacker can't always know the target's network topology, configurations and more information. 
Thus, an online approach is more suitable for attack situations.

Nirnay Ghosh et al. \cite{ghosh2009quantitative} proposed a methodology for generating an optimal attack path in wireless networks, which allows the network administrator to detect the most risk-prone attack path, i.e., the path that  is most likely to be chosen by the attacker in order to gain a critical resource in the network. By achieving that, the security administrator can prioritize the security mechanism. In this article, they used PSO (particle swarm optimization) in order to find the optimal attack path using attack vector metrics.  Each wireless node considers as a PSO particle and the motion of each particle replicate the mobility of the wireless nodes. Like the particles in PSO, the attacker's aim is always directed towards the target. Unlike the classic PSO, every node in the wireless have been assigned with weight that represent a severity measures obtained from customized risk parameters, also, the position of the node is immaterial.  
An artificial case study was presented, which modeled network with 100 hosts. After a pair of source and destination is given, the execution time grows exponentially as the number of visited nodes (=n) increases. Which means the time complexity is $O(c^n)$. It is important to mention that this is the first time such technique where the warm optimization concept is used to generate an attack path.

Another work of Nirnay Ghosh et al. is netsecuritas which \cite{ghosh2015netsecuritas} presents a heuristic-based attack graph generation algorithm which takes into consideration both attack graph properties and predefined security budget of any organization. The attack graph type represented is an exploit dependency graph – an acyclic directed graph which constructed from two types of nodes -  exploit and security condition and two types of edges – one represents a relation of require (condition to exploit) and the second is an imply relation (from exploit to condition node). It takes a list of exploits and a list of hosts as input, and generates attack graph in the form of an adjacency list. They use a backward chaining approach in order to find paths that terminate at the initial conditions. The algorithm starts from the goal state and then finds a list of exploits that can yield the desired condition. From this list of exploits, the algorithm makes a decision to choose the easiest one. The easiest one is choose based on the assumption that when multiple vulnerabilities exist in a host, the attackers choose the one easiest to exploit. The other exploits pushed into a stack. Now the algorithm finds a list of hosts through which that exploit can be executed, from this list it chooses the easiest to exploit and the others push into the stack. The easiest host is selected under the assumption that the host with the higher number of vulnerabilities is more likely to be detected and used by the attacker to execute the exploit. When the exploration is completed, successfully (reaches the initial condition) or in the dead end, the algorithm backtracks and evaluates other hosts or exploits from the stack. The presented algorithm is working at $O(n*e)$ when n is the number of hosts and e is the total number of exploits.

Penetration testing is a common methodology for identifying security risks, by running friendly attacks\cite{sarraute2013pomdps}. Simulated penetration testing automates this process, by designing a model of the system in question, and using model-based attack planning to generate the attacks. Basic methods from this domain were discussed in previous parts of this literature survey, and include, among others, the automatic generation of attack graphs by MulVAL, or model checking methods. Those models require complete knowledge of the network, hosts configuration, connectivity, and more for them to work. This level of knowledge of the system attacked is rarely available in the process of hacking a system. Ideally, a simulated penetration testing would conduct its attacks in a way similar to that of a real attacker, including the information gathering phase, and reasoning. Hence the goal of simulated penetration testing is much more ambitious: to realistically simulate a human hacker. In a more functional view, the simulated penetration testing model space spans a broad range of sequential decision making problems. In his work, Jorg Hoffmann \cite{sarraute2013penetration} has  analyzed prior work in AI and other relevant areas, and conducted a systematization of this model space, highlighting a multitude of interesting challenges to AI sequential decision making research.

Jorg \cite{sarraute2013penetration} Identified two major dimensions characterizing this space:
\begin{enumerate}[(A)]
\item  How to handle the uncertainty from the point of view of the attacker.
\item The degree of interaction between individual attack components
\end{enumerate}
To illustrate the point of (A), we need to remember that in many cases, exploits can fail, depending on the specific configuration needed for the exploit to run successfully, or protection measures taken by the target. To illustrate the point of (B) we should consider the implications of an attacker's action on the environment, such as crashing a machine after a failed exploitation attempt, or the ability to change network configuration. In the realm of uncertainty, there are 3 main models:
\begin{enumerate}
\item No uncertainty - each action succeeds
\item MDP - Markov Decision Process
\item POMDP - Partially Observable Markov Decision Process
\end{enumerate}
MDP in the context of attack graphs is a model in which each action has a number of possible outcomes, depending on a probability assigned to this action. Many models take the simplifying assumption that there are only two possible outcomes to each action (exploitation) which are success / failure. Of course in practice, this is not always the case. POMDP could be seen as an abstraction of MDP, in the sense that in POMDP, along with the basic actions of an MDP, there are sensing actions, which can alter the probability of the basic actions.\\
Until this day, modeling attack graphs in MDP or POMDP has been mainly theoretical and suffered from severe scalability issues. Very similar to the problem the academic world has encountered trying to solve attack graphs using model checkers. Relevant research done in this area could be seen in Jorg's previous research \cite{sarraute2013penetration} about POMDP in penetration testing, and Durkota's research \cite{durkota2014computing} work about modeling attack graph as MDP. Although both approaches shown above seem very interesting, they still suffer from scalability problems, and require a significant amount of simplifying assumptions that are not always possible. One of the reoccurring simplifying assumption is the knowledge of the network connectivity \cite{yuen2015automated}. The actions an attacker would take to discover the network could be used in various scenarios to detect an attacker in the network. 

Durkota et al. \cite{durkota2014computing} describes a method for representing an attack graph as an MDP. In this representation, each state in the MDP is a set of performed actions (exploits) and labels whether the action was successful or not. Each action is assigned a cost and a probability. Then, an optimized MDP solver is being used to narrow down the search space of the optimal policy. Some of the optimization techniques described are branch and bounds, which is used to prune unfavorable branches of the tree. Another optimization technique Durkota is offering is sibling-class theorem, which is used to identify multiple actions which results in the same effect. In these cases, only the lowest cost action is investigated further.
Compared to basic MDP solvers, Durkota's algorithm performed faster, and used less memory. Although, even with the optimization, the algorithm often ran out of memory, and it is assumed it will not scale well to large enterprise networks.  

Malte Helmert \cite{helmert2006fast} published the Fast Downward planner, which is a classical planning system based on heuristic search. It can deal with general deterministic planning problems encoded in the propositional fragment of PDDL2.2, including advanced features like ADL conditions and effects and derived predicates (axioms). Like other well-known planners such as HSP and FF, Fast Downward is a progression planner, searching the space of world states of a planning task in the forward direction. This planner can be used in the cyber field by converting any attack problem to pddl problem and solve it using this planner. 
\\We used this tool combined with the Jork Huffman POMDP converter \cite{sarraute2013penetration} and some guidelines from Durkota et al work \cite{durkota2014computing} in order to achieve "a worthy opponent" to our method.

%%.................................................................
\subsection{Deception}
%Throughout history the use of deception has evolved in societies around the world, and today, its use has expanded and it has become an integral part of our technical systems. \\
In cyber security, deception is often used to create a controlled path for attackers to follow starting from the the initial reconnaissance phase. 
%Several empirical experiments have been conducted in order to demonstrate the effect of deception in Cyber security. 
For example, Cohen and Koike~\cite{cohen2003leading,cohen2004misleading} showed how deception can control the path of an attack using red teams in experiments attacking a computer network.
Repik~\cite{repik2008defeating} makes a strong argument in favor of planned actions taken to mislead hackers. 
If the deception is obvious to the adversary, even unintentionally, the attacker can avoid, bypass, and even overcome the deceptive traps. We aim to avoid this situation and make the deception hard to identify by an adversary.
\note{RP}{Move to intro}

The recognition that deception could be an integral part of the network defense field, resulted in the need for a tool for deception, and Fred Cohen designed the \emph{Deception ToolKit (DTK)} \cite{DTK}, a tool which makes it appear to attackers as though the system running the DTK has a large number of widely known vulnerabilities. When the attacker issues a command or request, the DTK generates a predefined response, in order to encourage the attacker to continue its exploration of the host, or results in a shutdown of the service.

\subsubsection{Honeypots} 
Honeypots were first used in computer security in the late 1990s. 
Honeypots tempt the attackers to believe that they are valuable assets. 
However, being thoroughly monitored and usually isolated, honeypots give the defender the ability to detect and deflect the attack attempts~\cite{almeshekah2015using}:
\par Detection -- Honeypots result in a low false positive rate, because they are not intended to be used as part of the user's routine tasks in a system; thus any interaction with the honeypots is illegitimate.
\par Prevention -- Honeypots slow down attackers or discourage them from carrying on the attacks. For example, the LaBrea's "sticky" honeypots~\cite{StickyHoneypot} which answer connection attempts in a way that causes the machine on the other end to get "stuck", sometimes for a very long time.
\par Intelligence -- Honeypots are also used for gathering information about attacks. The Honeynet Project \cite{Honeynet} is an international security research organization, which invests its resources in the investigation of the latest attacks and the development of open source security tools to improve Internet security.

Naturally adversaries try to avoid honeypots.%, since honeypots may expose their attack methods. 
Rowe et al.~\cite{rowe2007defending} suggest the idea of fake honeypots, in which a system might pretend to be a honeypot in order to scare away attackers, reducing the number of attacks and their severity.

Fake vulnerabilities employed in current work for attack graph obfuscation transform a regular operational machine into a kind of high-interaction honeypot. 
However, there are several significant differences between fake vulnerabilties and honeypots. 
First, fake vulnerabilities must not allow exploitation of the vulnerability because they are installed on real operational systems. 
In contrast, exploitation is desirable in case of honeypots. 
Second, \note{RP}{TBD} Our approach does not alter the system while honeypotbased solutions introduce vulnerable machines in order to either capture the attacker or collect information for forensic purposes. 
Instead, we aim at deceiving attackers by manipulating their view of the target system and forcing them to plan attacks based on inaccurate knowledge, so that the attacks will likely fail.

\subsubsection{Virtual Attack Surface}

In \cite{albanese2016deceiving}, the authors propose algorithmic solutions to two classes of problems:
\begin{enumerate}
\item Inducing an external view that is at a minimum distance from the internal view, while minimizing the cost for the defender.
\item Inducing an external view that maximizes the distance from the internal view, given an upper bound on the cost for the defender.
\end{enumerate}
In their solution they build a manipulation graph - a graph that gathered \note{RP}{reflects?} all the possible manipulations to the current system. Then they present two heuristics in order to solve the problems described above:
\begin{enumerate}
\item TopKDistance - which recursively traverses the subgraph to
find a solution.
\item TopKBudget - answer the same problems as TopKDistance but improve time efficiency in the resolution.
\end{enumerate}
Furthermore, they present deception-based techniques for defeating an attacker’s effort to fingerprint operating systems and services
on the target system. They do it by modifying the outgoing traffic.

In contrast to considering only the impact of the manipulations on the defender, we consider both, the impact on the defender and on the attacker. Formulating the problem as an attack graph, enables to model the effect on the attacker's efforts for all possible paths to reach its goal in the network. In addition,  the impact of the fake vulnerabilities on the defender network is modeled as well.  

In our work, the main role of deception is to make the attacker believe that the information obtained is real. 
When this is successfully accomplished, the defender can gain a significant advantage over the adversary.

\subsection{Attack Graph games}
Recently there has been significant interest in applying game theory approaches to security. Durkota et al. introduced the term, "attack graph game" \cite{durkota2015optimal}, and presented a new leader-follower game-theory model of the interaction between a network administrator and an attacker who follows a multistage plan to attack the network. In order to determine the best strategy for the defender, they used the Stackelberg game (a two phase game), in which the defender, (the leader of the game), takes actions in order to strengthen the network by adding honeypots. Then, the attacker selects an optimal attack plan based on knowledge about the defender's strategy. The Stackelberg equilibrium is found by selecting the pure action of the defender that minimizes the expected loss under the assumption that the attacker will respond with an optimal attack.\\
The researchers presented the problem by using a type of logical attack graph that they refer to as dependency attack graphs which are generated by MulVAL. They model the problem as an MDP problem and assign random costs and rewards for solving it.\\
\par  In contrast to our research, there is a choice of how much decoys to put in the network in order to minimize the expected loss of the defender by increasing the probability that the attacker will be caught. Our main goal is to make it more difficult for the attacker, causing attrition and a waste of the adversary's valuable resources until the adversary is detected or waste its resources.\\
The results of this paper were based on an experiment conducted on a small business network (20 hosts) with 14 honeypots. They reported on performance of less than 10 seconds, but did not demonstrate it or analyzed the scalability for a larger network. Our method was tested on a real enterprise network and can scale for large networks as well.\\ 

\subsection{Active Defense}
Cohen and Koike presented a set of experiments where they successfully induced skilled red-team attackers to attack the targeted system in a particular sequence \cite{cohen2003leading}. Their main goal was to mimic physical attack tactics where such techniques can be used to lure the adversaries into specific zones by influencing their decisions, that means taking a specific path desired by the defenders.

Another approach was presented by Trassare \cite{trassare2013technique}.
He presented a technique to deceive attackers by giving them a fake internal network topology of the defender's choice. His strategy is to place the deceptive functionality on the border routers of an AS, because any designated device within the AS or the whole network is protected by the deception. Lastly, he presents a prototype implementation showing positive results.

Two Israeli startups: \emph{Cymmetria} \cite{Cymmetria} and \emph{TRAPX security} \cite{TRAPX} proposed practical ideas for integrating deception in the industry.

\emph{TRAPX} developed the DeceptionGrid Platform, each activity from the “lightest” reconnaissance to advanced breach attempts is contained, recorded and alerted enabling immediate remediation. It deploys automatic an array of decoys (Traps) and breadcrumbs (tokens) that provides visibility into ongoing attacks while luring attackers away from valuable assets.

\emph{Cymmetria} developed the MazeRunner, which creates breadcrumbs and decoys in order to lead attackers to believe that they have successfully gained access to a target machine. Having gained a false sense of security, attackers reveal their attack tools and methods, which defenders are then able to document and analyze. 

The above tools are using dedicated decoys, which are virtual machines (servers or other devices). These decoys are used to identify an interaction with the adversary based on the intuition that legitimate users are not expected to interact with these decoys. These decoys are reached by following breadcrumbs, found on an endpoint. Breadcrumbs are passive elements of data that can be found by the attacker in the reconnaissance phase (e.g. browser cookies, RDP and SSH credentials, shared folder mappings, OpenVPN scripts).

Our approach is different because we are not focusing on luring the attacker into specific points in the network, but focusing on causing the adversary significant loss of resources and  significant increase in the attack execution time.

\subsection{Obfuscation}
In order to perform a successful attack, one key piece of information that attackers need is the identity of the target system's running services and operating system. Operating systems have unique characteristics that can help an attacker to identify the operating system being used in the target system. Examples of these characteristics include the TCP/IP packet, response messages to queries, response messages to errors, predictability of sequence numbers, and banner information.\\
Murphy et al. \cite{murphy2010application} investigated the efficacy of using a host-based operating system (OS) obfuscation as an integral part of Air Force computer defenses. They used the \textit{OSfuscate} tool, by Crenshaw \cite{crenshaw2008osfuscate}, and concluded that it is effective in continuously obfuscating the host OS. OS obfuscation is a hard task! Making enough changes to normal OS functions that trick an attacker while maintaining a stable system that functions, is not a simple task.\\

\par \textbf{If the defenders use deception techniques, attackers can't make worry-free actions in the network. Instead, they are now forced to invest more resources, time, and effort in their attack attempts, and are under  constant fear to perform a wrong move and get caught. It is clear that deception creates a hostile environment for attackers.}\\

\par Unlike many studies in this field, our approach is not focusing on luring the attacker through specific paths. Our approach will focus on wasting the attacker's time and resources.

\subsection{Moving Target Defense}
A cyber moving target technique, is the presentation of a dynamic attack surface, increasing an adversary's work factor necessary to probe, attack, or maintain presence in a cyber target \cite{definedterm}.
Moving target defense (MTD) has been shown to be effective in cyber defense \cite{zhuang2012simulation} and may become a game changer in the cyber security arena.
In the following paragraphs, we will discuss some key examples from the MTD research field that is related to our research:\\

\par \textbf{Proactive Obfuscation :\\} In \cite{roeder2010proactive} the authors create server replicas that are likely to have fewer shared vulnerabilities. They do this by applying semantic-preserving code transformations which result diverse set of executables. In each period, they restart the servers with fresh images, these replicas will react differently to identical attacks. Thus, there are less opportunities for the attacker to compromise too many of the replicas that a service consists of. For example, the success of a buffer overflow attack typically will depend on stack layout details, so replicas using different obfuscated executables based on address reordering or stack padding are likely to crash instead of cave in to the attacker control.\\

\par \textbf{Dynamic Network Address Translation :}\\ In 1999, as part of the DARPA Information Assurance Program, a dynamic approach for active network defense was presented. This approach tried to verify the assumption that "Dynamic modification of defensive structures improves system assurance" \cite{kewley2001darpa}. The goal of this research was to prevent the attacker's ability to scan and map the network, which will make the attack more challenging. What they did is to change dynamically the addresses and port numbers used by the network's computers. The tool they used called Dynamic Network Translation (DYNAT), which camouflage the host identity information in TCP/IP packets. In their experiments they showed that their approach made it almost impossible to map the network while significantly increasing the attacker's effort. Beating DYNAT is difficult because in order to attack directly the target, the attacker needs to know the address hopping mechanism. However, there are number of drawbacks when using DYNAT: It requires that trusted computers computers on both sides of the communication be within the protection of DYNAT processes.\\

\par \textbf{Mutable Network :}\\ A Mutable Network (MUTE) \cite{al2011toward} is a method that composed changing network configurations such as: IP addresses, port numbers and routes between two points on the network. This method creates  kind of virtual overlay above the existing network, which all the traffic is routed in this virtual overlay. That means, that the original IP address and information on the systems never changes. This method uses encrypted channels in order to synchronize the IP address information.\\

\par \textbf{N-Variant Systems :} \\
The N-Variant systems framework \cite{cox2006n} executes a set of automatically diversified variants on the same inputs, and monitors their behavior to detect divergences. This platform based on the fact that an attacker requires to simultaneously compromise all system variants with the same input.
The proof-of-concept was built into the Linux kernel.\\

\par \textbf{Changing computer configurations using genetic algorithms :}\\
Crouse \& Fulp \cite{crouse2011moving} try to find more secure configurations of systems by using genetic algorithms. They use three ideas from genetics:
\begin{itemize}[noitemsep]
\item [$\circ$] Selection - selecting the best configurations based on their security score.
\item [$\circ$] Crossover - taking two configurations and combining elements of each one to create a new configuration.
\item [$\circ$] Mutation - randomly changing parts of a configuration to make it different from configurations on other systems.
\end{itemize}
The goal of this method is to create diverse configurations, temporally and spatially. Temporal diversity refers to the difference of configuration in a single computer over time and Spatial diversity refers to the difference of configuration between computers at any point in time.

The solution presented in the MTD field tends to reconfigure a system in order to modify the network, according to the adversary's perception. On the other hand, the view of the system is usually inferred by attackers based on the results of probing and scanning tools. Starting from this observation, we do not modify the network but rather provide fake responses to probing. %This  and probour goal is that when an attacker will construct an attack graph from the scans it collected, 
thus, the produced attack graph will be significantly different from the actual attack graph constructed from the network, without altering the system itself.

\par After a thorough review of studies that are directly related to our approach, we can say that to the best of our knowledge, there is no approach that adds false vulnerabilities and considers that the act of adding these "lies" has consequences to the routine operation of the network. Therefore, we were unable to find a method that carefully selects which fake vulnerabilities to add and indicates where to add them, optimally, in the network's PCs (without adding dedicated decoys), with the aims of making it \textbf{harder} for the attacker, forcing the attacker to use a significant amount of resources, and increasing the attack execution time dramatically.

%%%%%%%%%%%%%%%%%%%%%%%%%%%%%%%%%%%%%%%%%%%%%%%%%%%%%%%%%%%%%%%%%%%
%% Conclusions
%%%%%%%%%%%%%%%%%%%%%%%%%%%%%%%%%%%%%%%%%%%%%%%%%%%%%%%%%%%%%%%%%%%
\section{Conclusion and future work}\label{sec: Conclusions and future work}
\subsection{Discussion}
In this thesis, we gathered several guidelines for fake vulnerabilities assignments, which prevent detection of the applied deception by the adversary and also suggest how to assign fake vulnerabilities in an optimal manner.\\
When applying a deception mechanism it is important to mask the deception. If the deception mechanism is consistent and free from contradictions, it will be difficult to identify by the adversary and therefore its effect will be greater.\\
We present new way to defend networks by assigning fake vulnerabilities in an optimal way.\\
First, we model all the possible paths of an attack in the network using a state of the art model, called logical attack graph. Then, we find the optimal assignment of fake vulnerabilities in the given network, using the attack graph model. As initial results, we proposed a random assignment of fake vulnerabilities. In order to find an optimal assignment, that maximize the adversary's effort, we model this problem as an AI search problem. Where each node represents the fake assignments chosen so far, and the fake assignments left to choose.\\
We have proposed an admissible heuristic and a non-admissible heuristic, in order to examine the impact on the running time. We can see from the results that there was no significant difference in running time between the admissible heuristic and the non-admissible heuristic. Thus, in this case, an optimal assignment of fake vulnerabilities is preferable.\\
In order to reduce the algorithm's run time, besides examine an admissible and non-admissible heuristics, we offered a two more improvements: 
\begin{itemize}
\item Different search algorithms - We examine both A* and DFBnB search algorithms. From the results, we can see that in all tested cases, A* finds the optimal assignment of fake vulnerabilities and confirm it faster than the DFBnB search algorithm.
\item Sorting the list of assignment candidates - We examine two different list sorting - shortest path first and high utility first. Although there was no significant difference in run time between the two, but they were significantly better than a random sorting. Thus, the two suggested list ordering approaches improve the run time of the algorithm.
\end{itemize}

From the results described in the previous section, we can conclude that solving the problem of fake vulnerabilities assignment is applicable for enterprise networks, when using the appropriate parameters for search: \\
We can see that A* search algorithm can find and assert the optimal solution quickly, by using good ordering of the candidates list. 
For example, finding an optimal solution using A* search algorithm with shortest path first (or high utility first) candidate list sorting, can decreased up to 43 times less the run time factor, compared to find an optimal solution using DFBnB search algorithm with random ordering of the candidate list. \\
We found a good admissible heuristic, which does not dramatically increase running time, relative to the non-admissible one, and produce optimal results.\\ 
As described in this thesis, the input to this algorithm, besides the given enterprise network, is the budget. In a real functioning network, attackers may prefer to attack the network during the night, when the amount of communication passing through the network is lower. Furthermore, a defender may want, in order to not overload the network in a day time, to increase the amount of deception in the network. All of the above reinforces the choice of adjusting the amount of budget as input in the system.\\
Not less important, we confirm that our method is applicable in an enterprise network and adding fake vulnerabilities is feasible. We examine our approach on a real enterprise network with 80 hosts on a virtual machine with average capabilities at a reasonable time.

\subsection{Future work}
Our goal is to produce optimal deployment of fake vulnerabilities to an enterprise network.
The study can be continued in several directions:
First, we can and should shorten the running time. We can do it by reducing the size of the attack graph.
We can use \cite{gonda2017scalable}, \cite{noel2004managing}, \cite{homer2008improving} or \cite{zhang2011effective} works in order to reduce the given attack graph, and apply our AI search approach on it.\\
This will help us reduce the run time in two critical phases in our algorithm:
\begin{enumerate}
\item Minimize the number of potential deceptive computers.
\item We can reduce the time of finding the optimal path to the goal, at each fake vulnerabilities assignment, during the search algorithm. From the results we can see that this phase can repeat more than 10,000 times (see Figure \ref{fig:searchAlgorithmExpanded}).
\end{enumerate}
Another direction we can take, is to collect a real attacks data sets and model the attacker by using this knowledge. Our method can also fit after modifying the attacker's modeling.

\section*{References}

\bibliography{RP}

\begin{thebibliography}{10}
\expandafter\ifx\csname url\endcsname\relax
  \def\url#1{\texttt{#1}}\fi
\expandafter\ifx\csname urlprefix\endcsname\relax\def\urlprefix{URL }\fi
\expandafter\ifx\csname href\endcsname\relax
  \def\href#1#2{#2} \def\path#1{#1}\fi

\bibitem{murphy2010application}
S.~Murphy, T.~McDonald, R.~Mills, An application of deception in cyberspace:
  Operating system obfuscation1, in: International Conference on Information
  Warfare and Security, Academic Conferences International Limited, 2010, p.
  241.

\bibitem{kewley2001dynamic}
D.~Kewley, R.~Fink, J.~Lowry, M.~Dean, Dynamic approaches to thwart adversary
  intelligence gathering, in: DARPA Information Survivability Conference \&amp;
  Exposition II, 2001. DISCEX'01. Proceedings, Vol.~1, IEEE, 2001, pp.
  176--185.

\bibitem{huber2011host}
K.~E. Huber, Host-based systemic network obfuscation system for windows, Tech.
  rep., DTIC Document (2011).

\bibitem{rowe2007defending}
N.~C. Rowe, E.~J. Custy, B.~T. Duong, Defending cyberspace with fake honeypots,
  Journal of Computers 2~(2) (2007) 25--36.

\bibitem{sheyner2002automated}
O.~Sheyner, J.~Haines, S.~Jha, R.~Lippmann, J.~M. Wing, Automated generation
  and analysis of attack graphs, in: Security and privacy, 2002. Proceedings.
  2002 IEEE Symposium on, IEEE, 2002, pp. 273--284.

\bibitem{sheyner2004scenario}
O.~M. Sheyner, Scenario graphs and attack graphs, Ph.D. thesis, US Air Force
  Research Laboratory (2004).

\bibitem{ou2006scalable}
X.~Ou, W.~F. Boyer, M.~A. McQueen, A scalable approach to attack graph
  generation, in: Proceedings of the 13th ACM conference on Computer and
  communications security, ACM, 2006, pp. 336--345.

\bibitem{khaitan2011finding}
S.~Khaitan, S.~Raheja, Finding optimal attack path using attack graphs: a
  survey, International Journal of Soft Computing and Engineering 1~(3) (2011)
  2231--2307.

\bibitem{nessus}
R.~Deraison, "the nessus project", \url{http://www.nessus.org}.

\bibitem{OpenVAS}
{OpenVAS}, \url{http://www.openvas.org/}.

\bibitem{lyon2009nmap}
G.~F. Lyon, Nmap network scanning: The official Nmap project guide to network
  discovery and security scanning, Insecure, 2009.

\bibitem{artz2002netspa}
M.~L. Artz, Netspa: A network security planning architecture, Ph.D. thesis,
  Massachusetts Institute of Technology (2002).

\bibitem{ingols2006practical}
K.~Ingols, R.~Lippmann, K.~Piwowarski, Practical attack graph generation for
  network defense, in: Computer Security Applications Conference, 2006.
  ACSAC'06. 22nd Annual, IEEE, 2006, pp. 121--130.

\bibitem{williams2008garnet}
L.~Williams, R.~Lippmann, K.~Ingols, GARNET: A graphical attack graph and
  reachability network evaluation tool, Springer, 2008.

\bibitem{ou2005mulval}
X.~Ou, S.~Govindavajhala, A.~W. Appel, Mulval: A logic-based network security
  analyzer, in: USENIX Security Symposium, Vol.~8, Baltimore, MD, 2005.

\bibitem{NVD}
National vulnerability database, \url{https://nvd.nist.gov/}.

\bibitem{sagonas1994xsb}
K.~Sagonas, T.~Swift, D.~S. Warren, Xsb as an efficient deductive database
  engine, in: ACM SIGMOD Record, Vol.~23, ACM, 1994, pp. 442--453.

\bibitem{polad2017attack}
H.~Polad, R.~Puzis, B.~Shapira, Attack graph obfuscation, in: International
  Conference on Cyber Security Cryptography and Machine Learning, Springer,
  2017, pp. 269--287.

\bibitem{DTK}
F.~Cohen, {Deception Tool Kit }, \url{http://all.net/dtk/}.

\bibitem{albanese2016deceiving}
M.~Albanese, E.~Battista, S.~Jajodia, Deceiving attackers by creating a virtual
  attack surface, in: Cyber Deception, Springer, 2016, pp. 169--201.

\bibitem{puzis2007finding}
R.~Puzis, Y.~Elovici, S.~Dolev, Finding the most prominent group in complex
  networks, AI communications 20~(4) (2007) 287--296.

\bibitem{zhang1995performance}
W.~Zhang, R.~E. Korf, Performance of linear-space search algorithms, Artificial
  Intelligence 79~(2) (1995) 241--292.

\bibitem{stern2014max}
R.~T. Stern, S.~Kiesel, R.~Puzis, A.~Felner, W.~Ruml, Max is more than min:
  Solving maximization problems with heuristic search, in: Seventh Annual
  Symposium on Combinatorial Search, 2014.

\bibitem{sarraute2013penetration}
C.~Sarraute, O.~Buffet, J.~Hoffmann, Penetration testing== pomdp solving?,
  arXiv preprint arXiv:1306.4714.

\bibitem{sarraute2013pomdps}
C.~Sarraute, O.~Buffet, J.~Hoffmann, Pomdps make better hackers: Accounting for
  uncertainty in penetration testing, arXiv preprint arXiv:1307.8182.

\bibitem{helmert2006fast}
M.~Helmert, The fast downward planning system, Journal of Artificial
  Intelligence Research 26 (2006) 191--246.

\bibitem{DTKGuide}
S.~I. I.~R. Room, {Installing, Configuring, and Testing The Deception Tool Kit
  on Mac OS X },
  \url{https://www.sans.org/reading-room/whitepapers/detection/installing-configuring-testing-deception-tool-kit-mac-os-1056}
  (2006).

\bibitem{ammann2002scalable}
P.~Ammann, D.~Wijesekera, S.~Kaushik, Scalable, graph-based network
  vulnerability analysis, in: Proceedings of the 9th ACM Conference on Computer
  and Communications Security, ACM, 2002, pp. 217--224.

\bibitem{ghosh2009quantitative}
N.~Ghosh, S.~Nanda, S.~Ghosh, A quantitative approach towards detection of an
  optimal attack path in a wireless network using modified pso technique, in:
  Communication Systems and Networks and Workshops, 2009. COMSNETS 2009. First
  International, IEEE, 2009, pp. 1--10.

\bibitem{ghosh2015netsecuritas}
N.~Ghosh, I.~Chokshi, M.~Sarkar, S.~K. Ghosh, A.~K. Kaushik, S.~K. Das,
  Netsecuritas: An integrated attack graph-based security assessment tool for
  enterprise networks, in: Proceedings of the 2015 International Conference on
  Distributed Computing and Networking, ACM, 2015, p.~30.

\bibitem{durkota2014computing}
K.~Durkota, Computing optimal policies for attack graphs with action failures
  and costs, in: STAIRS 2014: Proceedings of the 7th European Starting AI
  Researcher Symposium, Vol. 264, IOS Press, 2014, p. 101.

\bibitem{yuen2015automated}
J.~Yuen, Automated cyber red teaming, Tech. rep., DTIC Document (2015).

\bibitem{cohen2003leading}
F.~Cohen, D.~Koike, Leading attackers through attack graphs with deceptions,
  Computers \& Security 22~(5) (2003) 402--411.

\bibitem{cohen2004misleading}
F.~Cohen, D.~Koike, Misleading attackers with deception, in: information
  assurance workshop, 2004. Proceedings from the fifth annual IEEE SMC, IEEE,
  2004, pp. 30--37.

\bibitem{repik2008defeating}
K.~A. Repik, Defeating adversary network intelligence efforts with active cyber
  defense techniques, Tech. rep., DTIC Document (2008).

\bibitem{almeshekah2015using}
M.~H. Almeshekah, Using deception to enhance security: A taxonomy, model, and
  novel uses, Ph.D. thesis, PURDUE UNIVERSITY (2015).

\bibitem{StickyHoneypot}
{LaBrea: "Sticky" Honeypot and IDS},
  \url{http://labrea.sourceforge.net/labrea-info.html}.

\bibitem{Honeynet}
{Honeynet Project}, \url{https://www.honeynet.org/}.

\bibitem{durkota2015optimal}
K.~Durkota, V.~Lis{\`y}, B.~Bo{\v{s}}ansk{\`y}, C.~Kiekintveld, Optimal network
  security hardening using attack graph games, in: Proceedings of IJCAI, 2015,
  pp. 7--14.

\bibitem{trassare2013technique}
S.~T. Trassare, A technique for presenting a deceptive dynamic network
  topology, Tech. rep., NAVAL POSTGRADUATE SCHOOL MONTEREY CA (2013).

\bibitem{Cymmetria}
{Cymmetria }, \url{https://www.cymmetria.com/}.

\bibitem{TRAPX}
{TRAPX security }, \url{https://trapx.com/product/}.

\bibitem{crenshaw2008osfuscate}
A.~Crenshaw, Osfuscate: change your windows os tcp/ip fingerprint to confuse
  p0f, networkminer, ettercap, nmap and other os detection tools, Retrieved Mar
  12 (2008) 2009.

\bibitem{definedterm}
\url{https://definedterm.com/moving_target_defense/}.

\bibitem{zhuang2012simulation}
R.~Zhuang, S.~Zhang, S.~A. DeLoach, X.~Ou, A.~Singhal, Simulation-based
  approaches to studying effectiveness of moving-target network defense, in:
  National symposium on moving target research, 2012, pp. 1--12.

\bibitem{roeder2010proactive}
T.~Roeder, F.~B. Schneider, Proactive obfuscation, ACM Transactions on Computer
  Systems (TOCS) 28~(2) (2010) 4.

\bibitem{kewley2001darpa}
D.~L. Kewley, J.~F. Bouchard, Darpa information assurance program dynamic
  defense experiment summary, Systems, Man and Cybernetics, Part A: Systems and
  Humans, IEEE Transactions on 31~(4) (2001) 331--336.

\bibitem{al2011toward}
E.~Al-Shaer, Toward network configuration randomization for moving target
  defense, in: Moving Target Defense, Springer, 2011, pp. 153--159.

\bibitem{cox2006n}
B.~Cox, D.~Evans, A.~Filipi, J.~Rowanhill, W.~Hu, J.~Davidson, J.~Knight,
  A.~Nguyen-Tuong, J.~Hiser, N-variant systems: a secretless framework for
  security through diversity, in: Usenix Security, Vol.~6, 2006, pp. 105--120.

\bibitem{crouse2011moving}
M.~Crouse, E.~W. Fulp, A moving target environment for computer configurations
  using genetic algorithms, in: Configuration Analytics and Automation
  (SAFECONFIG), 2011 4th Symposium on, IEEE, 2011, pp. 1--7.

\bibitem{gonda2017scalable}
T.~Gonda, R.~Puzis, B.~Shapira, Scalable attack path finding for increased
  security, in: International Conference on Cyber Security Cryptography and
  Machine Learning, Springer, 2017, pp. 234--249.

\bibitem{noel2004managing}
S.~Noel, S.~Jajodia, Managing attack graph complexity through visual
  hierarchical aggregation, in: Proceedings of the 2004 ACM workshop on
  Visualization and data mining for computer security, ACM, 2004, pp. 109--118.

\bibitem{homer2008improving}
J.~Homer, A.~Varikuti, X.~Ou, M.~A. McQueen, Improving attack graph
  visualization through data reduction and attack grouping, in: Visualization
  for computer security, Springer, 2008, pp. 68--79.

\bibitem{zhang2011effective}
S.~Zhang, X.~Ou, J.~Homer, Effective network vulnerability assessment through
  model abstraction, in: International Conference on Detection of Intrusions
  and Malware, and Vulnerability Assessment, Springer, 2011, pp. 17--34.

\end{thebibliography}

\end{document}